\documentclass[11pt]{article}
\usepackage[top=1.0in, bottom=1.0in, left=0.8in, right=0.8in]{geometry} 
\usepackage[T1]{fontenc}
\usepackage{textcomp}
\usepackage{dsfont}
\usepackage[frenchmath]{mathastext}  
\usepackage{mathrsfs}
\usepackage{algorithm}
\usepackage{algpseudocode}
\usepackage{amsmath,amssymb}
\usepackage{physics}
\usepackage{bm}
\usepackage{bbm}
\usepackage{graphicx} 										
\usepackage[font={small}]{caption}
\usepackage{booktabs}
\usepackage{multirow}
\usepackage{subfigure}
\usepackage{float}
\usepackage{upgreek}
\usepackage{enumerate}												
\usepackage{mdwlist}												
\usepackage[dvipsnames]{xcolor}											
\usepackage[style=apa, natbib, maxcitenames=2, maxbibnames=999, uniquelist=false, backend=biber]{biblatex}
\usepackage[plainpages=false, pdfpagelabels]{hyperref} 
	\hypersetup{
		colorlinks   = true,
		citecolor    = RoyalBlue,
		linkcolor    = RubineRed, 
		urlcolor     = RubineRed
	}               
\usepackage{ulem}
\usepackage{xcolor,mathtools}                      
\mathtoolsset{showonlyrefs=true}       
\usepackage{setspace,lipsum}

\usepackage{amsthm}
\usepackage{thmtools}            
	\allowdisplaybreaks                  
	\theoremstyle{plain}
	\newtheorem{theorem}{Theorem}
	\newtheorem{lemma}{Lemma}

	\theoremstyle{definition}

	\newtheorem{remark}{Remark}
	
	\newtheorem{hypothesis}{Hypothesis}
\newenvironment{customhypothesis}[1]
 {\renewcommand\thehypothesis{#1}\hypothesis}
  {\endhypothesis}

\usepackage{authblk}
\usepackage{soul}

\usepackage{empheq}

\usepackage{array}

\newcommand\Eb{\mathds{E}}
\newcommand\Fb{\mathds{F}}
\newcommand\Gb{\mathds{G}}

\newcommand\Ib{\mathds{1}}
\newcommand\Lb{\mathds{L}}
\newcommand\Pb{\mathds{P}}
\newcommand\Qb{\mathds{Q}}
\newcommand\Rb{\mathds{R}}

\newcommand\Ec{\mathcal{E}}
\newcommand\Fc{\mathcal{F}}
\newcommand\Gc{\mathcal{G}}
\newcommand\Hc{\mathcal{H}}

\newcommand\Nc{\mathcal{N}}

\newcommand\Qc{\mathcal{Q}}

\newcommand\Tc{\mathcal{T}}

\newcommand\Uc{\mathcal{U}}

\newcommand\tauh{\widehat{\tau}}

\newcommand\Wt{\widetilde{W}}
\newcommand\Xt{\widetilde{X}}
\newcommand\Yt{\widetilde{Y}}

\newcommand\Mt{\widetilde{M}}

\newcommand\mut{\widetilde{\mu}}

\DeclareMathOperator*{\arginf}{arg\,inf}

\definecolor{asparagus}{rgb}{0.0, 0.5, 0.0}
\definecolor{amaranth}{rgb}{0.9, 0.17, 0.31}

\newcommand{\eqlnostar}[2]{\begin{align}\label{#1}#2\end{align}}
\newcommand{\eqstar}[1]{\begin{align*}#1\end{align*}}
\newcommand{\eq}[1]{\ifthenelse{\equal{#1}{*}}
  {\eqstar}
  {\eqlnostar{#1}}
 }
\makeatother

\begin{document}

\title{Robust valuation and optimal harvesting of forestry resources in the presence of catastrophe risk and parameter uncertainty}

\author[a]{Ankush Agarwal}
\author[b,c,d]{Christian Ewald}
\author[b]{Yihan Zou\footnote{Corresponding author. {Emails: \url{aagarw93@uwo.ca} (A. Agarwal), \url{christian.ewald@inn.no} and \url{christian.ewald@umu.se} (C. Ewald), 
\url{yihan.zou@outlook.com} (Y. Zou).}}\thanks{Acknowledgements: The authors thank Kevin Kamm for his valuable assistance with model parameter estimation.}}

\affil[a]{\small Department of Statistical and Actuarial Sciences, University of Western Ontario, London, Canada.}
\affil[b]{\small Adam Smith Business School, University of Glasgow, Glasgow, United Kingdom.}
\affil[c]{\small University of Inland, Lillehammer, Norway.}
\affil[d]{\small School of Mathematics and Statistics, Umeå University, Umeå, Sweden.}

\date{\today}

\maketitle 

\vspace{20mm}
\begin{abstract}
\noindent We determine forest lease value and optimal harvesting strategies under model parameter uncertainty within stochastic bio-economic models that account for catastrophe risk. Catastrophic events are modeled as a Poisson point process, with a two-factor stochastic convenience yield model capturing the lumber spot price dynamics. Using lumber futures and US wildfire data, we estimate model parameters through a Kalman filter and maximum likelihood estimation and define the model parameter uncertainty set as the 95\% confidence region. We numerically determine the forest lease value under catastrophe risk and parameter uncertainty using reflected backward stochastic differential equations (RBSDEs) and establish conservative and optimistic bounds for lease values and optimal stopping boundaries for harvesting, facilitating Monte Carlo simulations. Numerical experiments further explore how parameter uncertainty, catastrophe intensity, and carbon sequestration impact the lease valuation and harvesting decision. In particular, we explore the costs arising from this form of uncertainty in the form of a reduction of the lease value. These are implicit costs that can be attributed to climate risk and will be emphasized through the importance of forestry resources in the energy transition process. We conclude that in the presence of parameter uncertainty, it is better to lean toward a conservative strategy reflecting, to some extent, the worst case than being overly optimistic. Our results also highlight the critical role of convenience yield in determining optimal harvesting strategies.

\smallskip

\noindent \textbf{Keywords}: Real options; model uncertainty; ambiguity; forests; catastrophe risk; climate risk; simulation.

\end{abstract}




\section{Introduction}
Conventional approaches to forest lease valuation and harvesting strategies often rely on the real options approach, which adapts financial option pricing methods to model investment timing under various types of risks (see, e.g., a review in \citet{dixit1994investment,nadarajah2023review,trigeorgis2018real}). While real options models are adequate for capturing investment flexibility under stochastic price dynamics, they typically assume stable parameters and overlook the importance of model parameter uncertainty. This limitation is particularly significant in the management of natural resources, where environmental and market conditions can be unpredictable, and catastrophic events, such as severe wildfires and storms, add layers of uncertainty to investment decisions. 

These uncertainties are further compounded by the increasing severity and frequency of climate events, now widely believed to be driven by human activities (\citet{IPCC2023}). However, what remains less clear is how this trend will progress and how these growing uncertainties will impact investment timing in natural resource management, particularly in decisions like the optimal harvesting of forest land. These compound events add complexities that extend beyond typical market fluctuations, necessitating a robust framework to assess and manage this increased uncertainty based on historical data. This need aligns with open problems 27-28 in operations research (OR) challenges for forestry raised by \citet{ronnqvist2015operations}. In this context, our work aims to provide a practical, data-driven approach that incorporates catastrophe risk and accounts for parameter uncertainty to capture the effects of climate-driven risks on optimal forest harvesting.

Carbon sequestration in forests plays a crucial role in the broader context of energy transmission and sustainability. Forests act as significant carbon sinks, absorbing carbon dioxide (CO2) from the atmosphere, which can mitigate climate change and contribute to energy systems that rely on renewable resources. The integration of forest carbon sequestration into energy policies can enhance the sustainability of energy transmission, particularly through the use of biomass as a renewable energy source. Restoring and conserving forests is recognized as a cost-effective strategy for reducing CO2 emissions, see \citet{10.1007/s44177-022-00028-y}. The carbon captured by forests can be utilized in various ways, including the production of bioenergy. This bioenergy, derived from forest biomass, is mostly considered carbon-neutral because the CO2 released during biomass combustion is offset by the CO2 absorbed during the growth of the trees, see \citet{10.4236/jep.2012.39114} and \citet{10.3390/su11030863}. Thus, the sustainable management of forests not only contributes to carbon sequestration but also supports the development of renewable energy sources that can be transmitted through existing energy infrastructures.

The economic implications of forest biomass for energy production are, in fact, very significant. Using forest residues and low-value wood can create additional income streams for forest owners, incentivizing them to maintain forest health and productivity, see \citet{10.1111/gcbb.12643}. This economic viability is crucial for promoting sustainable forest management practices that enhance carbon sequestration while simultaneously providing renewable energy resources. The potential of biomass energy, particularly when combined with carbon capture and storage (BECCS), offers a pathway to achieve negative emissions, which is vital for mitigating climate change, \citet{10.1002/ghg.34}.

In our paper, we account for the value of carbon sequestration and, therefore, implicitly the value that forests contribute to the energy transmission process and study how this impacts the optimal harvesting behaviors in the presence of parameter uncertainty and catastrophic risks. More generally, our work builds on and extends a body of real options literature addressing challenges in natural resource management.

Previous studies have applied the real options approach to model optimal harvesting strategies and forest valuation under stochastic price dynamics (e.g., \citet{insley2002real, kallio2012real}). Some studies considered the impact of price fluctuations alongside other stochastic factors, such as timber inventory (\citet{morck1989valuation}) and infestation rates (\citet{sims2011optimal}). However, only a few studies have explored the role of catastrophic events, such as wildfires, which can severely impact forestry investments. For example, \citet{reed1993decision} and \citet{yin1996effect} examined optimal forest harvesting under catastrophe risk but in simplified settings that allowed for analytical valuation results.

In addition, the real options approach has been used to evaluate a broader range of real-life investment projects under catastrophe risk. For instance, \citet{truong2016s} employed a net present value approach for evaluating projects aimed at reducing catastrophe risk, assuming climate change impacts are deterministic and fully known upfront. \citet{truong2018managing} extended this framework for climate adaptation projects by integrating climate change risk into investment decisions. \citet{kort2022preventing} investigated projects in which firms could influence the probability of catastrophic events. However, none of these models have examined the impact of catastrophe model parameter uncertainty on optimal investment timing and valuations.

Furthermore, our work aligns with real options studies that examine the impact of model uncertainty or ambiguity on optimal investment timing and project valuation. For example, \citet{nishimura2007irreversible} studied the optimal stopping problem for irreversible investments within a one-dimensional framework under various ambiguity settings. \citet{trojanowska2010worst} explored the optimal investment timing of a firm with stochastic profits under a specific type of model uncertainty known as $\kappa$-ignorance, while \citet{thijssen2011incomplete} addressed irreversible investment problems under maxmin utility over $\kappa$-ignorance in incomplete markets. In a related approach, \citet{cartea2017irreversible} studied robust indifference pricing of real options in incomplete markets, defining the ambiguity set with an entropic penalty function. \citet{hellmann2018fear} examined investment timing in a competitive context, showing how a firm without ambiguity aversion could leverage its rival’s ambiguity aversion under multiple-prior maxmin preferences. Another relevant study is \citet{loisel2022ambiguity}, which analyzed the classical Faustmann rotation problem of forest harvesting under storm risk and ambiguity. However, their work did not incorporate a continuous-time framework, omitted market risk, and focused only on scenario and frequency ambiguity of storms. Our study addresses these limitations by modelling forest harvesting decisions under catastrophe risk within a continuous-time framework that incorporates both market risk and parameter uncertainty.

Specifically, we model the arrival of catastrophes as a Poisson point process and employ a two-factor stochastic convenience yield model (\citet{schwartz1997stochastic}), renowned for its ability to produce realistic term structures for commodity futures prices. Using lumber futures data, we estimate the two-factor model with the Kalman filter and maximum likelihood estimator (MLE), following approaches similar to those in \citet{schwartz1997stochastic,trolle2009unspanned}. The Poisson jump intensity of catastrophic events is calibrated using data from major US wildfires by MLE. We then define the parameter uncertainty set as the 95\% confidence region inferred from these estimates. To assess the impact of statistical parameter uncertainty, we formulate forest lease values under catastrophe risk and parameter uncertainty as solutions to reflected backward stochastic differential equations (RBSDEs), mainly based on the theoretical result from \citet{el1997reflected}. Through Monte Carlo methods, we establish conservative and optimistic bounds for lease values and reveal the optimal stopping boundaries based on a modified version of the Stratified Regression One-step Forward Dynamic Programming (SRODP) algorithm of \citet{agarwal2023monte}. Additionally, we conduct numerical experiments to investigate how model uncertainty affects forest lease values, considering factors such as catastrophe jump intensity and the inclusion of carbon sequestration value in harvesting strategies. The results underscore the influence of convenience yield on optimal harvesting decisions and confirm that a conservative probability belief accelerates optimal harvesting while an optimistic probability belief delays it. Moreover, the inclusion of carbon sequestration value further postpones the optimal harvesting time. Finally, assessed against the case of no uncertainty, we investigate the loss of lease value that originates from the presence of parameter uncertainty. This reflects the true cost of uncertainty and can indirectly be attributed to climate risk. To do so, we evaluate the performance of the optimistic and conservative harvesting strategies under the market pricing measure and compare these against the ``true'' optimal strategy under the no uncertainty assumption.  

Our study makes several valuable contributions to the field of OR and the existing literature: 
\begin{enumerate}
    \item We develop a mathematical framework for addressing the practical problem of forest harvesting under catastrophe risk, incorporating statistical parameter uncertainty - a novel approach within the fields of real options and forestry management.
    \item We employ sophisticated mathematical tools, specifically reflected backward stochastic differential equations (RBSDEs), to answer important and relevant questions about forestry management and real options; establishing a new paradigm for addressing optimal stopping problems for irreversible investments within multi-dimensional stochastic frameworks.
    \item Our analysis is grounded in a rich, real-world dataset, leveraging market data to estimate the uncertainty region. Consequently, our findings reflect parameters for lumber price dynamics estimated from actual market data and catastrophe intensity parameters calibrated to historical data on severe disasters.
    \item We offer several insights for industry practitioners and forestry risk managers, demonstrating how parameter uncertainty influences optimal harvesting times and emphasizing the critical role of convenience yield in harvesting decisions. Additionally, numerical results confirm that incorporating carbon sequestration value delays the optimal harvesting time.
    \item We determine the optimal stopping boundary that triggers the harvesting decision under parameter uncertainty through numerical experiments, making our analysis particularly relevant for practical forestry management decisions.
\end{enumerate}

The remainder of the paper is organized as follows: Section \ref{sec: the model} presents the modelling framework, detailing the incorporation of catastrophe risk and parameter uncertainty in forest lease valuation and optimal harvesting. Section \ref{sec: data summary} outlines the methodology for parameter estimation and presents the estimation results. Section \ref{sec: numerical results} introduces the Monte Carlo method and discusses numerical results for robust optimal forest harvesting and valuation under conservative, no uncertainty and optimistic scenarios. Finally, Section \ref{sec: conclusion} discusses
the main conclusions of this work. Appendix \ref{appendix: algo} details our numerical algorithm for optimal harvesting and valuation. Appendix \ref{appendix: proofs} contains proofs for all theoretical results.

\section{Problem formulation}\label{sec: the model}

We are concerned with harvesting a publicly owned forest for a single rotation from the perspective of a social planner or agent. The agent is assumed to have a leasehold on forest land with merchantable trees whose wood volume grows (in weight) as time progresses. The agent decides the time to harvest the forest and has the objective to maximize profit from the sale of lumber. There is a harvesting cost associated with cutting trees as well as a loss of future amenity benefits (net of management costs). Moreover, the agent loses potential increase in wood volume if the trees are harvested. The agent not only bears market risk due to uncertain lumber prices but is also exposed to the occurrence of catastrophic disasters, e.g., wildfires or storms. We thus propose a mathematical framework for modelling stochastic lumber prices and the intensity of the occurrence of catastrophic events. We use a popular, deterministic growth function to model the biological growth of the forest.
\subsection{Dynamics under the real-world probability measure}\label{sec: dynamics under the real-world probability measure}
Classical approaches to model lumber prices include the use of mean-reverting diffusion processes (\citet{insley2002real}) and multi-factor models (\citet{morck1989valuation}). Here, we use a two-factor stochastic convenience yield model of \citet{schwartz1997stochastic}. Let the state vector $X: = (\delta, P)^\top$ denote the collection of net convenience yield and lumber spot price, respectively. In this context, net convenience yield refers to the benefit of holding lumber inventories minus storage costs. A positive convenience yield indicates that the holding benefit exceeds the storage expense. For simplicity, we refer the net convenience yield as convenience yield hereafter. We assume that the dynamics of the state vector $X$ under the real-world probability measure  $\Pb$ are given as
\eqlnostar{eq: state vector 1}{
\dd X_t & = \mu^X(t,X_t)\dd t + \Sigma^X(t, X_t)\dd W_t,
}
where $\mu^X$ and $\Sigma^X$ are $\Rb^{2}$ and $\Rb^{2\times 2}$-valued functions. $\mu^X$, $\Sigma^X$ and $W_t$ are defined as
\eqstar{
\mu^X(t,X_t) = \begin{pmatrix} \kappa^{\delta}(\mu^{\delta} - \delta_t) \\ \mu^P(\delta_t)P_t\end{pmatrix},\, \Sigma^X(t,X_t) = \begin{pmatrix} \sigma^{\delta} & 0\\ \rho\sigma^P P_t & \sqrt{1-\rho^2}\sigma^P P_t
\end{pmatrix},\,W_t = \begin{pmatrix}W^1_t\\W^2_t\end{pmatrix}.
}
In the above, $W = (W^1, W^2)^\top$ denotes a standard two-dimensional Brownian motion under $\Pb$. $\kappa^{\delta}, \mu^{\delta}, \sigma^{\delta}$ and $\sigma^P$ are constants, $\mu^P(\cdot)$ is an $\Rb$-valued function, and $\rho\in (-1,1)$ is also a constant. In addition, the risk-free money market account $B$ is defined as $B_{s,t} := \exp((t-s) r), s\leq t$ where $r$ is the constant risk-free interest rate. Assume that under $\Pb$ we have a right-continuous filtration $\Fb = (\Fc_t)_{0\leq t\leq T},$ generated by $W$ which drives the market movements, but excludes any information on natural catastrophes. Hence, the filtration $\Fb$ contains only the market information, defining the probabilistic setup of the market $(\Omega,(\Fc_t)_{0\leq t\leq T},\Pb)$. 

Next, we describe the details of modelling the catastrophe risk using an intensity-based approach. 
We use a random time $\xi: \Omega \rightarrow [0, +\infty),$ announcing the event of a catastrophe and define the process $N_t := \Ib_{\{\xi\leq t\}}$, as a step process taking a value of zero on the event set $\{\xi>t\}$ and one otherwise. We denote by $\Gb = (\Gc_t)_{t\geq 0},$ the minimal enlargement of the filtration $\Fb$ such that $\xi$ is a $\Gb$-stopping time. Generally, the random time $\xi$ is not necessarily an $\Fb$-stopping time, which means that the agent might not know when the catastrophe occurs by observing information from the financial market only. However, we assume that the agent observes $\Gb$, that is, knows when the catastrophe has happened. Furthermore, the following well-known hypothesis \ref{hypothesis H} is assumed to hold throughout this article (see \citet{elliott2000models} or \citet{bielecki2004credit} for a detailed discussion in the context of credit risk).
\begin{customhypothesis}{{\bf (H)}}\label{hypothesis H}
Every $\Fb$-square integrable martingale is a $\Gb$-square integrable martingale.
\end{customhypothesis}
Hypothesis \ref{hypothesis H} implies that the $\Fb$-Brownian motion $W$ remains a Brownian motion on $\Gb$ and thus will justify our change of measure in Section \ref{sec: Q measure}. We further assume that the step process $N$ has a $\Gb$-compensator under $\Pb$ with an intensity $\lambda^{\Pb}$ which is an $\Fb$-predictable density process. This means that the occurrence of the catastrophe is modeled by a first jump of an (inhomogeneous) Poisson process $N$ with stochastic intensity $\lambda^{\Pb}\in(0,\infty)$. The stopping time $\xi$ is then an exponentially distributed random variable conditioned on the market information $\Fc_t$. Therefore, we can denote the compensated Poisson process $M$ of the step process $N$ by
\eqlnostar{}{M_t = N_t - \int_0^{t\wedge \xi}\lambda^{\Pb}_s\dd s, \;\;\text{for}\;\; 0\leq t\leq T,}
which is a $\Gb$-martingale under the measure $\Pb$. Naturally, the conditional exponential distribution suggests the following representation for the continuous conditional probability of the event $\{\xi > t\}$:
\eqlnostar{}{\Pb(\xi > t|\Fc_t) = \exp\left(-\int_0^t \lambda^{\Pb}_s\dd s\right),\,\,\text{for}\,\,0\leq t\leq T.}

\subsection{Valuation of futures and forestry investments under a risk-neutral measure}\label{sec: Q measure}
Agents will decide to harvest the forestry resource by considering its valuation at any time $t$. Valuations will be carried out under a measure that reflects relevant hedging opportunities, diversification, as well as relevant risk premia. This is, in principle, achieved by any risk-neutral equivalent martingale measure (EMM). As our market model proposed in Section \ref{sec: dynamics under the real-world probability measure} is intrinsically incomplete due to the presence of stochastic convenience yield, multiple EMMs can exist, reflecting the agents' beliefs. Thus, to complete our probabilistic structure, we first characterize the set $\Qc$ of equivalent martingale measures (EMMs). As usual, we say that a probability measure $\Qb$ is an EMM, i.e. $\Qb\in\Qc$, if and only if $\Qb$ is equivalent to the reference real-world measure $\Pb$ and the discounted prices of tradable assets under $\Qb$ are martingales. We will work under the assumption that hypothesis \ref{hypothesis H} holds for any $\Qb\in\Qc$. Through a Girsanov-type argument as used in \citet{kusuoka1999remark}, we obtain $\Qb$ as
\eqlnostar{eq: RN derivative 1}{\frac{\dd \Qb}{\dd \Pb}\Big|_{{\Gc}_t} = \eta_t^{\theta,\phi},\,\,\forall E\in\Gc_t,}
and its associated density process $\eta^{\theta,\phi}$ satisfies $\eta^{\theta,\phi}_0=1.$ Moreover,
\eqlnostar{}{\dd \eta_t^{\theta,\phi} = \eta_{t-}^{\theta,\phi}\left({\theta_t}^\top\dd W_t + (\phi_t - 1)\dd M_t\right),\,\,\text{for}\,\,0\leq t\leq T,}
where $\theta = (\theta^1, \theta^2)^\top$ and $\phi\in(0,\infty)$, satisfy technical conditions such that $\eta^{\theta,\phi}$ is a $\Gb$-adapted strictly positive martingale. The model dynamics under $\Qb$ are then given as
\eqlnostar{eq: state vector 2}{
\dd X_t & = \mut^X(t,X_t)\dd t + \Sigma^X(t, X_t)\dd \Wt_t,}
where $(\Wt_t)_{0\leq t\leq T}$ is a two-dimensional $\Qb$-Brownian motion and
\eqlnostar{}{
\mut^X(t,X_t) = \mu^X(t,X_t) + \Sigma^X(t,X_t) \theta_t = \begin{pmatrix} \kappa^{\delta}(\mut^{\delta} - \delta_t) \\ (r - \delta_t)P_t\end{pmatrix} 
, \,\,\Wt_t = \begin{pmatrix}\Wt^1_t\\ \Wt^2_t\end{pmatrix} = \begin{pmatrix} W^1_t-\theta_t^1\\ W^2_t-\theta_t^2\end{pmatrix},}
with the ``risk-neutral'' constant parameter $\mut^{\delta} = \mu^{\delta} + \sigma^{\delta}\theta^1_t/\kappa^{\delta}$. The compensated Poisson process $\Mt$ given as
\eqlnostar{}{\Mt_t = N_t - \int_0^{t\wedge\xi}\phi_s\lambda_s^{\Pb}\dd s,\,\,\text{for}\,\,0\leq t\leq T,}
is a $\Qb$-martingale, and $\lambda^{\Qb} = \phi\lambda^{\Pb}$ is the corresponding $\Qb$-intensity of the step process $N$.
\begin{remark}
Technically, the condition $\phi\in(0,\infty)$ justifies the change of measure in \eqref{eq: RN derivative 1} and guarantees that $\Qb$ is equivalent to $\Pb$. Moreover, it also means that the probability of the occurrence of a potential catastrophe is always positive in a ``risk-neutral'' world since $\lambda^{\Qb} = \psi\lambda^{\Pb}$ and the assumption $\lambda^{\Pb}\in(0,\infty)$.
\end{remark}
Notably, since no-arbitrage only requires that the condition $\mu^P(\delta_t) + \rho\sigma^P\theta^1_t + \sqrt{1-\rho^2}\sigma^P\theta^2_t = r - \delta_t,\,t\in[0, T]$ holds and otherwise $\theta$ can take arbitrary form, the EMM is not necessarily unique in our case, i.e. the market is intrinsically incomplete. Moreover, the market price of catastrophe risk $\phi\lambda^{\Pb}$ is also not reflected in the spot price, thus, $\phi$ is also undetermined. 
Given any EMM $\Qb$, we know that futures prices for maturity $T$ determined under $\Qb$  
\eqlnostar{eq: futures price 1}{F(t, T) := F(P_t,\delta_t,r,t,T) = \Eb^\Qb\left(P_T|\mathcal{F}_t\right)}
are consistent with the no-arbitrage condition, i.e., the corresponding market, which includes the trading of futures with the prices specified as above, remains arbitrage-free. With the introduction of futures contracts, the market also becomes complete, and a single EMM $\Qb$ is designated as the market measure, i.e., the EMM adopted by the market. Prices of all assets need to be determined under the market measure, which we will henceforth simply denote as $\Qb$. That said, due to statistical uncertainty, the agents in the model may not know the market measure exactly and apply other EMMs to determine their optimal harvesting policies and their own subjective values of the forestry resource. While futures prices can be observed on the market, it is not always possible to extract the pricing measure $\Qb$ exactly from the futures prices; at the least, there will be a statistical error, resulting in uncertainty for the agents, and the possibility of a range of EMMs from the agents' perspectives, as we will specify later.

Nevertheless, the ``true'' value of the forest lease is equivalent to the value of a financial option obtained under the market measure $\Qb$, where exercising corresponds to harvesting the forest. We define the value process as follows:
\eqlnostar{eq: value function 1}{
V_t := \Ib_{\{\xi\leq t\}}\zeta_{\xi} + \underset{\tau\in\Tc_{t,T}}{\sup}\, \Eb^{\Qb}\left(B_{t,\tau}^{-1}\left(P_{\tau}G_{\tau} - K\right)\Ib_{\{\tau<\xi\}} + \int_{t\wedge\xi}^{\tau\wedge\xi} B_{t,u}^{-1} A_u\dd u\bigg| \Gc_t\right),}
where $\zeta$ is a bounded, $\Fb$-adapted continuous process. 
$\Tc_{t, T}$ denotes the set of stopping times with respect to the filtration $(\Fc_u)_{t\leq u\leq T}$ with values in $[t, T]$, $K$ is the fixed harvesting cost, $A$ is an $\Rb$-valued $\Fb$-predictable process representing the amenity value (net of management costs) that flows from the forest in period $t$ if the trees are not harvested, and $G$ is a deterministic, age-dependent growth function of the timber weight per hectare. The next result evaluates the worth of forest land up until the moment a catastrophe occurs.
In the event that the catastrophe has not happened yet, that is, $\Ib_{\{\xi>t\}}$, we notice that the following is true for a given $\tau \in \Tc_{t,T}$ by following Lemma 3.1 in \citet{elliott2000models} 
\eqstar{
&\Ib_{\{\xi>t\}}\Eb^{\Qb}\left(B_{t,\tau}^{-1}\left(P_{\tau}G_{\tau} - K\right)\Ib_{\{\tau<\xi\}} + \int_{t\wedge\xi}^{\tau\wedge\xi} B_{t,u}^{-1} A_u\dd u\bigg| \Gc_t\right)\\
&= \Ib_{\{\xi>t\}}\Eb^{\Qb}\left(B_{t,\tau}^{-1}\left(P_{\tau}G_{\tau} - K\right)\Gamma^{\lambda}_{t, \tau} + \int_{t}^{\tau} B_{t,u}^{-1}\Gamma^{\lambda}_{t,u} A_u\dd u\bigg| \Fc_t\right),
}
where $\Gamma^{\lambda}_{s,t}:=\Qb(\xi>t|\xi>s,\Fc_s) = \exp\left(-\int_s^t \lambda^{\Qb}_u\dd u\right).$ The above follows since $G$ is deterministic, $P$ and $A$ are $\Fb$-adapted, and $\xi$ is an exponentially distributed random variable conditioned to $\Fc_t$. Thus, going forward we focus on computing 
\eqlnostar{eq: F-reduced Q value}{
\sup_{\tau \in \Tc_{t,T}}\Eb^{\Qb}\left(B_{t,\tau}^{-1}\left(P_{\tau}G_{\tau} - K\right)\Gamma^{\lambda}_{t, \tau} + \int_{t}^{\tau} B_{t,u}^{-1}\Gamma^{\lambda}_{t,u} A_u\dd u\bigg| \Fc_t\right) =: v_t. 
}
Based on the definition above, we have the following relation $V_t \Ib_{\{\xi>t\}} = v_t.$ We provide a rigorous justification of this relationship in the technical Theorem \ref{theorem 2}.
\subsection{Framework for model uncertainty}\label{sec: uncertainty set}
After calibrating a stochastic model to derivative prices, we usually take the EMM implied by the calibrated model as the ``true'' pricing measure. However, the calibrated model parameters are either statistically uncertain or locally optimal. Thus, there is a natural model uncertainty when choosing a calibrated model as the ``true'' pricing measure since the ``true'' pricing measure is necessarily given by a set of equivalent martingale pricing measures. 
Let us fix the measure corresponding to parameter estimates as our reference pricing measure $\Qb$ and define model uncertainty relative to it using the confidence interval for the parameter estimates. To introduce model uncertainty, we assume that the model is subject to uncertainty in the ``risk-neutral'' drift and catastrophe intensity parameters $(\kappa^{\delta}, \mut^{\delta}, \lambda^{\Qb})$ through their $\Fb$-predictable replacements $ (\kappa^{\delta,u}, \mut^{\delta,u}, \lambda^{u})$, where $\lambda^{u}:[0, T]\times\Omega\rightarrow\Rb_+$. The ambiguous parameters are indexed by $(u_t)_{0\leq t\leq T}$ which is a stochastic control process that determines their values. The values of $(\kappa^{\delta,u}, \mut^{\delta,u}, \lambda^{u})$ are assumed to lie in an uncertainty set which is a hypercube $U$
\eqstar{U:= [\underline{\kappa},\overline{\kappa}]\times[\underline{\mu},\overline{\mu}]\times[\underline{\lambda},\overline{\lambda}]\subset\Rb_+\times\Rb\times\Rb_+.}
Then we can define a space $\Uc$ of admissible control processes as
\eqstar{\Uc := \left\{u\in \Lb^2\left(\Fb,\Rb^3\right): (\kappa_t^{\delta,u}, \mut_t^{\delta,u}, \lambda_t^{u})(\omega)\in U,\dd t\times\dd \Pb-a.s.\right\}
}

Each $u$ implies a different EMM $\Qb^u$ through the density generators $\alpha$ and $\psi,$ defined below. Like in \eqref{eq: RN derivative 1}, we define $\Qb^u$ with $\Qb$ as the reference measure by
\eqlnostar{eq: RN derivative 2}{\frac{\dd \Qb^u}{\dd \Qb}\Big|_{{\Gc}_t} = \eta_t^{\alpha,\psi},\,\,\forall E\in\Gc_t,}
with $\eta^{\alpha,\psi}$ satisfying $\eta^{\alpha,\psi}_0=1$ and
\eqlnostar{eq: density 2}{\dd \eta_t^{\alpha,\psi} = \eta_{t-}^{\alpha,\psi}\left({\alpha_t}^\top\dd \Wt_t + (\psi_t - 1)\dd \Mt_t\right),\,\,\text{for}\,\,0\leq t\leq T,}
where $\alpha := (\alpha^{u,1}, \alpha^{u,2})^\top$ and $\psi = \frac{\lambda^u}{\lambda^{\Qb}}\in(0,\infty)$. This means the process $\eta^{\alpha,\psi}$ 
can be explicitly expressed as ${\eta_t^{\alpha,\psi} 
= \eta^{\alpha}_t \eta^{\psi}_t,}$
with
\eqlnostar{}{\eta^{\alpha}_t  = 
\exp\left(\int_0^t{\alpha_k}^\top\dd \Wt_k - \frac{1}{2}\int^t_0\langle\alpha,\alpha\rangle_k\dd k\right), && 
\eta^{\psi}_t =
\exp\left(\int_{(0, \, t]}\ln\psi_k\dd N_k - \int_0^{t\wedge\xi}\lambda^u_k\dd k\right).}
The model dynamics under $\Qb^u$ are then given as
\eqlnostar{eq: state vector 3}{
\dd X_t & = \mut^u(t,X_t)\dd t + \Sigma^X(t, X_t)\dd \Wt^u_t,}
where $\left(\Wt^u_t\right)_{0\leq t\leq T}$ is a two-dimensional $\Qb^u$-Brownian motion and the no-arbitrage principle implies
\eqlnostar{}{
\mut^u(t,X_t) = \mut^X(t,X_t) + \Sigma^X(t,X_t) \alpha_t = \begin{pmatrix}\kappa^{\delta,u}_t\left(\mu^{\delta,u}_t - \delta_t\right)\\ (r - \delta_t)P_t\end{pmatrix}, \,\,\Wt^u_t = \begin{pmatrix}\Wt^{u,1}_t\\ \Wt^{u,2}_t\end{pmatrix} = \begin{pmatrix} \Wt^1_t-\alpha_t^{u,1}\\ \Wt^2_t-\alpha_t^{u,2}\end{pmatrix}.}
Simple calculations yield that the pair $(\alpha^{u,1},\alpha^{u,2})$ in the density process $\eta^{\alpha,\psi}$ should satisfy
\eqlnostar{eq: alpha condition}{\alpha_t = \begin{pmatrix}\alpha^{u,1}_t\\ \alpha^{u,2}_t
\end{pmatrix} = \begin{pmatrix}\frac{\kappa^{\delta,u}_t \mu^{\delta,u}_t - \kappa^{\delta}\mu^{\delta} - \left(\kappa^{\delta,u}_t - \kappa^{\delta}\right)\delta_t}{\sigma^{\delta}}\\ \frac{- \rho\left(\kappa^{\delta,u}_t \mu^{\delta,u}_t - \kappa^{\delta}\mu^{\delta} - \left(\kappa^{\delta,u}_t - \kappa^{\delta}\right)\delta_t\right)}{\sqrt{1-\rho^2}\sigma^{\delta}}
\end{pmatrix},\,\,\text{for}\,\,t\in[0,T].}

Intuitively, the above equality in conjunction with $\psi = \frac{\lambda^u}{\lambda^{\Qb}}$ provides a connection between the uncertain parameters $(\kappa^{\delta,u}, \mut^{\delta,u}, \lambda^{u})$ and the set $\Qc^U$ of EMMs under uncertainty through the density process defined in \eqref{eq: density 2}. Thus, we can write $\alpha_t$ and $\psi_t$ as $\alpha(t,u_t)$ and $\psi(t,u_t)$ for every $t\in[0, T]$ relating the control process $(u_t)$ with the density generators. This implies that each EMM $\Qb^u$ is associated with each $u$, that is, choice of parameters $(\kappa^{\delta,u}, \mut^{\delta,u}, \lambda^{u})$. 
Since each agent may have a different choice of uncertain parameters within their discretion, we will henceforth refer to $\Qb^u$ as the subjective measure reflecting the corresponding agent's subjective beliefs. Valuation of the forest lease under $\Qb^u$ will result in a price different from the price obtained under the market measure $\Qb$ and a different optimal harvesting policy. In the following, we will refer to the lease values obtained under $\Qb^u$ with $u \neq 0$ as subjective values. 



\subsection{Conservative and optimistic valuation under model uncertainty}\label{sec: robust valuation}
The previous continuous conditional probability of the event $\{\xi > t\}$ can be defined under $\Qb^u$ as 
\eqlnostar{}{\Gamma^{u}_{s,t}:=\Qb^u(\xi>t|\xi>s,\Fc_s) = \exp\left(-\int_s^t \lambda^{u}_k\dd k\right).}
Without the knowledge of actual market pricing measure, the agent, will evaluate whether to harvest at time $\tau$ before the catastrophic disaster happens, under their subjective EMM $\Qb^u$. Then based on the result in  \eqref{eq: F-reduced Q value}, we define the $\Fb$-reduced subjective value function of the forest lease under $\Qb^u$ as
\eqlnostar{eq: F-reduced Qu value}{v^{\Qb^u}_t := \underset{\tau\in\Tc_{t,T}}{\sup}\, \Eb^{\Qb^u}\left(B^{-1}_{t,\tau}\Gamma^{u}_{t,\tau}\left(P_{\tau}G_{\tau} - K\right) + \int_t^{\tau}B_{t,k}^{-1}\Gamma^{u}_{t,k}A_k\dd k\Big|\Gc_t\right),\,\,\text{for}\,\,t\in[0,T].}
Thus, for every choice of $(\kappa^{\delta,u}, \mut^{\delta,u}, \lambda^{u})$, we obtain a different value of the forest lease, reflecting the effect of model uncertainty on the lease valuation. In our framework of model uncertainty, we are able to connect $v^{\Qb^u}_t$ defined above to the solution of a reflected backward stochastic differential equation (RBSDE) through the following result.
\begin{lemma}\label{lemma: Qu value}
\begin{enumerate}
    \item[(i)] The process $v^{\Qb^u}$ given by \eqref{eq: F-reduced Qu value} satisfies, for every $t\in[0,T]$,
    \eqlnostar{eq: F-reduced Qu value under Q}{v^{\Qb^u}_t = \underset{\tau\in\Tc_{t,T}}{\sup}\, (\eta^{\alpha}_t)^{-1}\Eb^{\Qb}\left(B^{-1}_{t,\tau}\Gamma^{u}_{t,\tau}\eta^{\alpha}_{\tau}\left(P_{\tau}G_{\tau} - K\right) + \int_t^{\tau}B^{-1}_{t,k} \Gamma^{u}_{t,k}\eta^{\alpha}_{k}A_k \dd k\bigg|\Fc_t\right).}
    \item[(ii)] Let the $\Fb$-progressively measurable triplet $\{(Y^u,Z^u,K^u), 0\leq t\leq T\}$ taking values in $\Rb\times\Rb^2\times\Rb_+$ be a unique solution to the RBSDE
    \eqlnostar{eq: rbsde Qu 1}{Y^u_t = S_{T} + \int_t^T f\left(s,X_s,Y^u_s,Z^u_s,u_s\right)\dd s + K^u_T - K^u_t - \int_t^T {Z^u_s}^{\top}\dd \Wt_s,\;\text{for}\;t\in[0,T],}
    where $S_t := S(t,X_t) = P_t G_t - K$, $f\left(t,x,y,z,u\right) = A_t - (r+\lambda^u) y + \alpha^{u,1} z^1 + \alpha^{u,2} z^2$ with $x := (x^1,x^2)^\top$ and $z := (z^1,z^2)^\top$, $K^u$ is a non-decreasing, continuous process with $K_0^u = 0$ ensuring $Y^u_t\geq S_t, \,t\in[0,T]$ with $\int_0^T(Y^u_t - S_t)\dd K_t^u = 0$. Then, for every $t\in[0,T]$, $Y^u_t = v^{\Qb^u}_t$, $\Qb$-a.s.
\end{enumerate}
\end{lemma}

In Lemma \ref{lemma: Qu value}(i), we re-express $v^{\Qb^u}_t$ under the reference measure $\Qb,$ which allows to derive the result in Lemma \ref{lemma: Qu value}(ii). Thus, for every choice of $(\kappa^{\delta,u}, \mut^{\delta,u}, \lambda^{u})$ coming from the model uncertainty framework, we can solve for the forest lease value function by solving an RBSDE like the one stated in \eqref{eq: rbsde Qu 1}. However, since the set of admissible control processes $\Uc$ contains infinite possibilities, we instead focus on computing the so-called \textit{optimistic value} $v^+$ and \textit{conservative value} $v^-$. We define these to be the upper and lower subjective value of the forest lease before the catastrophe under parameter uncertainty described by $\Uc.$ To simplify notations, we will denote $v^{\Qb^u}$ and $\Eb^{\Qb^u}(\cdot)$ by $v^{u}$ and $\Eb^{u}(\cdot)$, respectively. Then for every $t\in[0,T]$, the \textit{optimistic value} and \textit{conservative value} are defined as
\eqlnostar{eq:upper lower val}{v_t^+ := \underset{(u_t)\in \Uc}{\sup}\, v^{u},\quad\quad
v_t^- := \underset{(u_t)\in \Uc}{\inf}\, v^{u}.}
Since each choice of $(u_t)$ leads to a different EMM $\Qb^u$, one can naturally interpret $v^+$ as the valuation under an ``optimistic'' probability belief since we consider the supremum of the value functions $v^{u}$ over all the possible choices for the control process $u$, whereas $v^-$ can be seen as the valuation under a ``conservative'' probability belief since we instead consider the infimum of the $v^{u}$ over all the possibilities for $u$. These functions are based on the agent's beliefs and do not reflect the market values of the leases managed according to the corresponding harvesting strategies. The latter will be discussed in Section \ref{sec: Forest lease values under full parameter uncertainty}. The two value functions $v^+$ and $v^-$ and their corresponding optimal stopping strategies can be found by solving the appropriate RBSDEs as shown below.
\begin{theorem}\label{thm: robust rbsde}
Let the $\Fb$-progressively measurable triplets $\{(Y^+,Z^+,K^+), 0\leq t\leq T\}$ and $\{(Y^-,Z^-,K^-), 0\leq t\leq T\}$  taking values in $\Rb\times\Rb^2\times\Rb_+$ be unique solutions to the RBSDE \eqref{eq: rbsde Qu 1} with generator $f$ replaced with 
\eqlnostar{}{f^+\left(t,x,y,z\right) = \underset{u\in U}{\sup}\, f\left(t,x,y,z,u\right) = A_t + \underset{u\in U}{\sup}\, \left(-(r - \lambda^u) y + \alpha^{u,1} z^1 + \alpha^{u,2} z^2\right),} and \eqlnostar{}{f^-\left(t,x,y,z\right) = \underset{u\in U}{\inf}\, f\left(t,x,y,z,u\right) = A_t + \underset{u\in U}{\inf}\, \left(-(r - \lambda^u) y + \alpha^{u,1} z^1 + \alpha^{u,2} z^2\right),} respectively. Then $Y^+_t = v_t^+$ and $Y^-_t = v_t^-$ for every $t\in[0, T]$ $\Qb$-a.s.
\end{theorem}
In the above, we slightly abuse the notation and denote different subjective values of the ambiguous parameters $(\kappa^{\delta,u}, \mut^{\delta,u}, \lambda^{u})$ by corresponding values of the control process $u.$ As mentioned earlier, the exact relationships between the two are as in \eqref{eq: alpha condition} and $\psi = \frac{\lambda^u}{\lambda^{\Qb}}.$ It is well-known (see \citet[Appendix D]{karatzas1998methods}) that the optimal stopping time in the context of our RBSDE setup of Lemma \ref{lemma: Qu value} satisfies
\eqlnostar{}{
{\tau^u} := \inf\{t\geq 0: Y^u_t = S_t\}.
}
This means that it is optimal for the agent to harvest the forest land as soon as the subjective lease value $Y^u_t$ equals the instant harvesting revenue $S_t$. To obtain the corresponding optimal stopping strategies for the optimistic and conservative cases, we can solve the RBSDEs stated in Theorem \ref{thm: robust rbsde}. Thus, we have
\eqlnostar{eq:opt stop time}{
{\tau^+} := \inf\{t\geq 0: Y^+_t = S_t\}, && {\tau^-} := \inf\{t\geq 0: Y^-_t = S_t\}.
}
Based on the optimal stopping strategies $\tau^+$ and $\tau^-$, we compute the following two quantities
\eqlnostar{eq:lease value}{
\Eb^{\Qb}\left(B_{0,\tau^{\pm}}^{-1}\left(P_{\tau}G_{\tau^{\pm}} - K\right)\Gamma^{\lambda}_{0, \tau^{\pm}} + \int_{0}^{\tau^{\pm}} B_{0,u}^{-1}\Gamma^{\lambda}_{0,u} A_u\dd u\right).
} 
$\tau^{\pm}$ can be either $\tau^+$ or $\tau^-.$ The above are the ``true'' forest lease values obtained under the market measure when the harvesting strategies  ${\tau^+}$ and ${\tau^-}$ reflect subjective beliefs. We compare the above two lease value estimates with the valuation computed by directly solving the RBSDE \eqref{eq: rbsde Qu 1} corresponding to $u \equiv 0$, that is, without any model uncertainty. By looking at the three values - market valuation with an optimistic harvesting strategy, market valuation with a conservative harvesting strategy, and market valuation with no model uncertainty, we can quantify the impact of parameter uncertainty in our model and its effect on the forest lease valuation. 

\section{Model estimation}\label{sec: data summary}
We use data on weekly lumber futures prices and historical wildfire frequency for the estimation of parameters in the two-factor stochastic convenience yield model and the catastrophe jump intensity of Section \ref{sec: Q measure} via maximum likelihood estimation (MLE). We define the parameter uncertainty set $U$ as the Cartesian product of the 95\% confidence regions of the estimated parameters. The set $U$ also defines the space $\Uc$ of control processes $(u_t)$. 
We use the Kalman filter and MLE to estimate the parameters in the two-factor stochastic convenience yield model using lumber futures prices, as the lumber spot price and convenience yield are not directly available in the market. For the catastrophe jump intensity $\lambda^\Qb$, we assume that it is a constant parameter and that there is no jump risk premium, i.e. $\lambda^\Pb = \lambda^\Qb$. We then estimate this parameter using MLE.

\subsection{Estimation of the stochastic convenience yield model}\label{sec: estimate two-factor model}
We use weekly data on six of the most liquid lumber futures contracts with time-to-maturities of up to six months traded at the Chicago Mercantile Exchange (CME). Our data span the period between 8th September 1993 and 27th June 2022 (1520 weekly observations).
The risk-free interest rate $r = 0.0231$ is chosen as the average 6-month Treasury Bill rate for this period. We denote the six different futures contracts by F1 through F6 representing various average time-to-maturities (0.042 years to 0.459 years). Table \ref{table: futures data} provides summary statistics of prices and time-to-maturities for these futures. It can be observed that time-to-maturities are stable and remain within a narrow range for each contract throughout the sample period.
\begin{table}[H]
\centering
\begin{tabular}{ccc}
\hline
& Mean price & Mean time-to-maturity \\
Contract & (Standard deviation) & (Standard deviation) \\ \hline
F1 & 347.49 (173.43) & 0.042 (0.024) \\
F2 & 346.32 (164.74) & 0.125 (0.024) \\
F3 & 346.80 (156.51) & 0.209 (0.024) \\
F4 & 345.71 (145.11) & 0.292 (0.024) \\
F5 & 346.48 (139.86) & 0.376 (0.024) \\
F6 & 345.58 (130.66) & 0.459 (0.024) \\ \hline
\end{tabular}
\caption{{\footnotesize Statistics of futures contracts from September 1993 to June 2022. We use F1 as the futures closest to time-to-maturity and F6 furthest to time-to-maturity.}}
\label{table: futures data}
\end{table}

We estimate the model parameter vector $\Theta := \{\sigma^P,\sigma^{\delta},\kappa^{\delta},\mut^{\delta},\rho\}$ using the Kalman filter and MLE. As mentioned before, the lumber spot price $P$ and convenience yield $\delta$ are latent state variables. Following \citet{schwartz1997stochastic} and \citet{trolle2009unspanned}, we cast the two-factor model into state space form, which includes a transition equation and a measurement equation. For this purpose, we create an equidistant time grid $\pi=(t_i)_{i=0,1,...,N}$ with step size $\Delta t$, where $t_0 = 0$ and $t_N = T$. The transition equation here is essentially equivalent to the discretized version of the forward process in \eqref{eq: state vector 2} using the Euler scheme. The discretized forward process is given as
\eqlnostar{eq: discrete state vector}{X_{t_{i+1}}^\pi := \begin{pmatrix}\delta_{t_{i+1}}^\pi\\ P_{t_{i+1}}^\pi\end{pmatrix} = X_{t_{i}}^\pi + \begin{pmatrix} \kappa^{\delta}(\mut^{\delta} - \delta_{t_i}^\pi) \\ (r - \delta_{t_i})P_{t_i}^\pi\end{pmatrix}\Delta t + \begin{pmatrix} \sigma^{\delta} & 0\\ \rho\sigma^P P_{t_i}^\pi & \sqrt{1-\rho^2}\sigma^P P_{t_i}^\pi
\end{pmatrix} \begin{pmatrix}\Delta \Wt^1_{t_i}\\ \Delta \Wt^2_{t_i}\end{pmatrix},} where the $(i+1)$-th Brownian motion increment under measure $\Qb$ is denoted as $\Delta \Wt_{t_i} = \Wt_{t_{i+1}} - \Wt_{t_{i}}$ with $i = 0, \ldots, N-1.$ For the transition equation in the Kalman filter, let us define $\Xt_{t_i}^\pi := (\delta_{t_i}^\pi,p_{t_i}^\pi)^\top$ where $p_{t_i}^\pi := \log(P_{t_i}^\pi)$. Then, we have the following
\eqlnostar{eq: transition equation}{\Xt_{t_{i+1}}^\pi = \begin{pmatrix} \kappa^{\delta}\mut^{\delta} \\ r - \frac{1}{2}{\sigma^P}^2\end{pmatrix}\Delta t + \begin{pmatrix} 1 - \kappa^{\delta}\Delta t & 0\\ -\Delta t & 1
\end{pmatrix}\Xt_{t_i}^\pi + \eta_i,\, \text{where}\, \eta_i\sim \Nc (0,E),\,
E := \begin{pmatrix} {\sigma^{\delta}}^2 & \rho {\sigma^{\delta}}{\sigma^{P}}\\ \rho {\sigma^{\delta}}{\sigma^{P}} & {\sigma^{P}}^2
\end{pmatrix},}
and $\eta_i$ is serially uncorrelated. The measurement equation characterizes the relationship between state variables and logarithmic futures prices $F_{t_i} := (F_{t_i}^j)_{j=1,...,6}$, as we categorize them into six different contracts. Under the dynamics \eqref{eq: state vector 2}, the futures price $F(t,T)$ at time $t$ for a contract expiring at time $T$ is given as
\eqlnostar{eq: two factor pricing equation}{F(t,T) = & P_t\exp\left(-\delta_t\cdot\frac{1-\exp\left(-\kappa^{\delta}(T-t)\right)}{\kappa^{\delta}} + D(T-t)\right),\\
\label{eq: two factor pricing equation D}
D(T-t) := & \left(r - \mut^{\delta} + \frac{1}{2}\left(\frac{\sigma^{\delta}}{\kappa^{\delta}}\right)^2 - \frac{\sigma^P\sigma^{\delta}\rho}{\kappa^{\delta}}\right)(T-t) + \frac{1}{4}{\sigma^{\delta}}^2\frac{1 - \exp\left(-2\kappa^{\delta}(T-t)\right)}{{\kappa^{\delta}}^3} \\
& + \left(\mut^{\delta}\kappa^{\delta} + \sigma^P\sigma^{\delta}\rho - \frac{{\sigma^{\delta}}^2}{\kappa^{\delta}}\right)\frac{1 - \exp\left(-\kappa^{\delta}(T-t)\right)}{{\kappa^{\delta}}^2}.}
 Thus, the measurement equation writes
\eqlnostar{eq: measurement equation}{F_{t_i}^j = D\left(\Delta T^j\right) + \begin{pmatrix}\frac{-\left(1-\exp\left(-\kappa^{\delta}\Delta T^j\right)\right)}{\kappa^{\delta}}\\ 1\end{pmatrix}\cdot \Xt_{t_i}^\pi + \varepsilon_i^j,\, j=1,...,6,}
where $D(\cdot)$ is defined in \eqref{eq: two factor pricing equation D}, $\varepsilon_i := (\varepsilon_i^j)_{j=1,...,6} \sim \Nc(0,\text{diag}(d_1^2,...,d_6^2))$ is serially uncorrelated and independent of $\eta$, and $\Delta T^j, j=1,...,6$ are time-to-maturities of six different contracts. 

Now we can estimate $\Theta$ by maximizing the likelihood that $(F_{t_i})$ matches the market log-futures prices while simultaneously estimating the filtered unobservable spot price and convenience yield. The second column in Table \ref{table: two-factor model estimation} presents the estimation results for $\Theta$ and the standard deviation parameters appearing in the measurement equation \eqref{eq: measurement equation}. It can be observed that almost all parameters are highly significant, except for the long-run average convenience yield $\mut^{\delta}$. 
\citet{schwartz1997stochastic} also found that the long-run average convenience yields related to gold and oil in the two-factor and three-factor models are mostly insignificant.

To further examine the relevance of the parameter $\mut^{\delta}$, we manually set it to zero and re-estimate the two-factor model without it. The estimation results are included in the third column of Table \ref{table: two-factor model estimation}. 
We note that the remaining parameter estimates and the maximized log-likelihood are nearly identical in both models. This means that the parameter estimation is stable and not overly sensitive to the inclusion or exclusion of the insignificant parameter. Since we only consider uncertainty in drift parameters, we define the uncertainty set for $\kappa^\delta$ via the 95\% confidence interval $[\underline{\kappa},\overline{\kappa}] = [0.8183,1.0699]$. The unchanged log-likelihood suggests that the simplified model (with one fewer parameter) has essentially the same explanatory power as the original model. The value of $\mut^{\delta}$ being close to zero also indicates that it has a negligible effect on the likelihood function and does not contribute meaningfully to the model. However, simply removing this parameter from the original two-factor model could have a substantial effect on the valuation of forestry investments and, consequently, the optimal harvesting decision of a forest lease. To handle this situation, we specify the uncertainty set for $\mut^{\delta}$ to be the range of values such that when we fix $\mut^{\delta}$ equal to any value in that range, the estimates for other parameters in the $\Theta$ vector remain stable. 
This process 
gives us an uncertainty set of $\mut^{\delta}$,  $[\underline{\mu},\overline{\mu}] = [-0.1020,0.0090]$.

\begin{table}[H]
\centering
\begin{tabular}{lcc}
\hline
Parameter              & Estimates with $\widetilde{\mu}^\delta$ & Estimates without $\widetilde{\mu}^\delta$ \\ \hline
$\sigma^P$             & 0.3304 (0.0061)***                      & 0.3304 (0.0060)***                         \\
$\sigma^\delta$        & 0.4640 (0.0131)***                      & 0.4640 (0.0131)***                         \\
$\kappa^\delta$        & 0.9441 (0.0642)***                      & 0.9441 (0.0641)***                         \\
$\widetilde\mu^\delta$ & 0.0000 (0.0206)                         & -                                          \\
$\rho$                 & 0.7061 (0.0156)***                      & 0.7061 (0.0153)***                         \\
$d_1$                  & 0.0364 (0.0005)***                      & 0.0364 (0.0005)***                         \\
$d_2$                  & 0.0150 (0.0003)***                      & 0.0150 (0.0003)***                         \\
$d_3$                  & 0.0238 (0.0005)***                      & 0.0238 (0.0004)***                         \\
$d_4$                  & 0.0137 (0.0003)***                      & 0.0137 (0.0003)***                         \\
$d_5$                  & 0.0224 (0.0004)***                      & 0.0224 (0.0004)***                         \\
$d_6$                  & 0.0083 (0.0003)***                      & 0.0083 (0.0003)***                         \\
Log-likelihood         & 20230.2                                 & 20230.2                                    \\ \hline
\end{tabular}
\caption{{\footnotesize Parameter estimates for the two-factor stochastic convenience yield model, with standard errors in parentheses. $[***]$ stands for statistical significance at 1\% level.}}
\label{table: two-factor model estimation}
\end{table}

\subsection{Estimation of the catastrophe jump intensity}\label{sec: estimation of catastrophe jump intensity}
We also estimate the jump intensity $\lambda^\Qb$ using MLE. We model it as the intensity of a Poisson process determining the number of monthly jumps of the step process $(N_t)$, which announces the arrival of catastrophes. The Federal Emergency Management Agency (FEMA)\footnote{\url{https://www.fema.gov/data-visualization/disaster-declarations-states-and-counties}} in the USA provides county-wise overall loss assessment of catastrophes, including severe local storms, wildfires, flooding, tornadoes, and winter weather.
We chose the data from Oregon as it is one of the largest forestry industrial regions in the USA and assume that the forestry investment is based in Douglas County, one of the top lumber-producing counties. Since wildfires are generally considered the most hazardous disaster for the forest, we obtained relevant wildfire data for Douglas County from FEMA, ranging from May 1953 to May 2024. 
The estimated yearly jump intensity $\lambda^\Qb$ is 0.2392 with a standard error of 0.0580, which is highly significant. Therefore, we take the 95\% confidence interval of $\lambda^\Qb$ as its uncertainty set, i.e. $[\underline{\lambda},\overline{\lambda}] = [0.1255,0.3528]$. 

After the estimation exercise, we have all the parameters $\Theta$ under the $\Qb$ measure, as shown in the second column of Table \ref{table: two-factor model estimation}. Moreover, the uncertainty set is given as $U = [\underline{\kappa},\overline{\kappa}]\times[\underline{\mu},\overline{\mu}]\times[\underline{\lambda},\overline{\lambda}]$ as $[0.8183,1.0699]\times[-0.1020,0.0090]\times[0.1255,0.3528]$. These estimates are used in the numerical experiments, the results of which are presented in the following section.

\section{Numerical experiments}\label{sec: numerical results}
As shown in Theorem \ref{thm: robust rbsde}, the optimistic and conservative values of the forestry investment are solutions to their corresponding RBSDEs.
Since these RBSDEs do not permit solutions in analytical forms, we evaluate the forestry investment and optimal harvesting strategy by solving the RBSDEs numerically. We adapt the Stratified Regression One-step Forward Dynamic Programming (SRODP) algorithm of \citet{agarwal2023monte} for solving RBSDEs for our purpose of forest lease valuation. The SRODP method numerically solves RBSDEs of the type in \eqref{eq: rbsde Qu 1} by using the classical sampling technique of stratified sampling. For more details on stratified sampling and its use in numerical methods for solving BSDEs we refer to \citet{gobet2016stratified}. Here, we provide the main idea of the numerical procedure and include details on the modification of the SRODP method in Appendix \ref{appendix: algo}.

Let us consider the same time discretization $\pi$ as used in Section \ref{sec: estimate two-factor model} for model estimation, and the discretized forward process $X$ defined as in \eqref{eq: discrete state vector}. Our aim then is to first compute the optimal stopping strategies defined in \eqref{eq:opt stop time} by numerically solving the RBSDEs stated in Theorem \ref{thm: robust rbsde}. For the time discretization $\pi,$ we consider discretized versions of the RBSDEs 
\eqlnostar{eq:EulerDis1}{Y_{t_i}^{\pi} = Y_{t_{i+1}}^{\pi} + f^\pm(t_i, {X_{t_i}^{\pi}}, Y_{t_i}^{\pi}, Z_{t_i}^{\pi})\Delta t + \Delta K_{t_i}^{\pi} - {Z_{t_i}^{\pi}}\cdot \Delta \Wt_{t_i},}
where $f^\pm$ can be either $f^+$ or $f^-$. By taking a conditional expectation in \eqref{eq:EulerDis1} without taking the term $\Delta K_{t_i}$ into account, we obtain the following
\eqlnostar{EulerDis3}{
\Yt_{{t_i}}^{\pi} = \Eb_{t_i}^{\Qb}(Y_{t_{i+1}}^{\pi} + f^{\pm}(X_{t_{i}}^{\pi},\Yt_{t_{i}}^{\pi},Z_{t_{i}}^{\pi})\Delta {t}),
}
where we use the short-hand notation $\Eb_{t_i}^{\Qb}(\cdot)$ for $\Eb^{\Qb}({\cdot|\Fc_{t_i}})$. 
This defines the expected value $\Yt^{\pi}$ (also known as the \textit{continuation value}) of the BSDE before reflection. Next, we consider an approximation \eqref{eqSRODP2} replacing $\Yt_{t_{i}}^{\pi}$ with $\Yt_{t_{i+1}}^{\pi}$ in the argument of $f$ in \eqref{EulerDis3} for simplicity.
{$K_{t_i}^{\pi}$ pushes $Y_{{t_i}}^{\pi}$ above the obstacle $S^\pi_{t_i}:= P_{t_i}^{\pi}G_{t_i}-K.$ Therefore, by taking the maximum in \eqref{EulerDis3} between the reflection bound $S^\pi_{t_i}$ and the expected value of the BSDE before reflection $\Yt_{{t_i}}^{\pi}$, leads to the approximation of $Y_{t_i}^{\pi}$ in \eqref{eqSRODP3}.} We thus get the following recursive procedure to solve for $Y^{\pi},$
\eqlnostar{eqSRODP1}{
    & Y_{t_N}^{\pi} = P^{\pi}_{t_N}G_{t_N}-K, \: Z_{t_{i}}^{\pi} = \frac{1}{\Delta t}\Eb_{t_i}^{\Qb}(Y_{t_{i+1}}^{\pi}\Delta \Wt_{t_i}), \\
    \label{eqSRODP2}
    & \Yt_{{t_i}}^{\pi} = \Eb_{t_i}^{\Qb}(Y_{t_{i+1}}^{\pi} + f^{\pm}(X_{t_{i}}^{\pi},Y_{t_{i+1}}^{\pi},Z_{t_{i}}^{\pi})\Delta {t}), \\
    \label{eqSRODP3}
    & {Y}_{t_{i}}^{\pi} = \Yt_{t_{i}}^{\pi}\vee S_{t_i}.
}
Note that $\Delta K_{t_i}$ in \eqref{eq:EulerDis1} does not disappear when computing $Y_{{t_i}}^{\pi}$, but instead is implicitly present in \eqref{eqSRODP3}. The SRODP algorithm solves the above system \eqref{eqSRODP1}-\eqref{eqSRODP3} by accurately approximating the conditional expectation in \eqref{eqSRODP1}-\eqref{eqSRODP2}. Once we have obtained the numerical approximate solution $Y^{\pi}$ of the RBSDE, we can find the approximate optimal stopping strategy. We define the approximate optimal stopping (harvesting) time-based on optimistic and conservative valuations as
\eqlnostar{eq:est opt stop time}{
\tauh^+ := \min\{t_i, i=0,...,N: Y^{+,\pi}_{t_i} = S^\pi_{t_i}\}, && \tauh^- := \min\{t_i, i=0,...,N: Y^{-,\pi}_{t_i} = S^\pi_{t_i}\}.
}
The approximate optimal harvesting time in the case of no model uncertainty can also be defined analogously as follows
\eqlnostar{eq:est opt stop time nounc}{
\tauh := \min\{t_i, i=0,...,N: Y^{\pi}_{t_i} = S^\pi_{t_i}\}.
}
Once we have estimated the three different optimal stopping strategies, we estimate the forest lease value under the market measure $\Qb$ by computing the following quantity
\eqlnostar{eq:lease value est}{
\Eb^{\Qb}\left(B_{0,\tau}^{-1}\left(P_{\tau}G_{\tau} - K\right)\Gamma^{\lambda}_{0, \tau} + \int_{0}^{\tau} B_{0,u}^{-1}\Gamma^{\lambda}_{0,u} A_u\dd u\right),
} 
for $\tau = \tau^+, \tau^-, \tauh$ using the Monte Carlo estimator. All numerical experiments have been implemented in Python on a personal computing device with a 5 GHz Intel Core i9-12900HX processor, 32 GB RAM and 12GB GeForce NVIDIA RTX A3000 GPU.

\subsection{Robust optimal harvesting strategies}\label{sec: robust optimal harvesting strategy}
\subsubsection{With full parameter uncertainty}\label{sec: harvesting strategy with full parameter uncertainty}
In this section, we consider full parameter uncertainty, meaning that the three risk-neutral parameters $(\kappa^{\delta}, \mut^{\delta}, \lambda^{\Qb})$ are replaced by $(\kappa^{\delta,u}, \mut^{\delta,u}, \lambda^{u})$, as defined in Section \ref{sec: uncertainty set}, and may take values within the set $U$. We first fix realistic parameterizations for the timber growth function $G$; harvesting cost $K$, and amenity value $A$. Since the underlying commodity for lumber futures on the CME is softwood lumber, such as pine and fir, we direct our attention to pine wood. Following \citet{insley2002real}, we use the logistic growth curve for merchantable white pine wood volume and assume parameter estimates as obtained by \citet{rollins1995financial}. The wood volume growth function $G$ (in cubic meters per hectare) is assumed to have the form
\eqlnostar{}{G_t = 792 - 5313 t^{-0.5},\, t\in[50,103].}
We assume that $G$ is 0 before the age of 50 years and becomes constant after the age of 103 years. Since wood volume is 0 before age 50, we also assume that a wildfire cannot affect the forest resources before that time. The harvesting cost $K = \$ 127.74$/1000 board feet is computed by subtracting the 2018 average stumpage price from the delivered price (including all costs associated with harvesting and transporting the timber) of pine sawtimber and is retrieved from \citet{bardon2023}. We assume that the amenity value $A$ is a constant with no management cost and that amenity solely comes from the annual hunting lease rate for pine stands. From the estimate in \citet{godar2022forest}, we use $A\equiv \$8$/hectare. 
In our numerical procedure explained in the previous section, we assess the subjective lease valuation of a one-hectare pine forest, with all calculations adjusted for different units so that the lease value is expressed in dollars per hectare. Following the approach in \citet{insley2002real}, we select similar parameters for the time horizon and discretization. Specifically, we set the maximum time horizon for harvesting to $T = 150$ years and choose $N = 3000$ steps, yielding a finer increment of $\Delta t = 0.05$ (approximately 18 days or half a month).
Thus, in our framework, the forest harvest can be decided every $\Delta t$ period. The initial lumber spot price is set at $P_0= \$ 600$ /per 1000 board feet, approximately matching the June 2022 lumber futures price with the nearest maturity.
The initial convenience yield is $\delta_0=-0.01$, chosen to be close to the estimated long-run mean $\widetilde{\mu}^\delta$ of the convenience yield. 
In the SRODP algorithm, we fix the space domain of the logarithmic lumber spot price $p$ to be $[-2.5,8.5]$, and that of the convenience yield $\delta$ to be $[-2,2]$. Furthermore, we consider $80$ hypercubes per dimension and 1000 simulated paths per hypercube, and conduct basis function regression within each hypercube using local polynomial functions of order one. These hyperparameters have already been fine-tuned for maximum accuracy and efficient performance, as discussed in \citet{agarwal2023monte}.


We estimate the optimal harvesting time under the optimistic, conservative, and no uncertainty scenarios as defined in \eqref{eq:est opt stop time} and \eqref{eq:est opt stop time nounc} by averaging the estimated stopping times for $10^5$ re-simulated paths and present the results in Table \ref{table: approximate optimal harvesting time}. 
Our estimated results align well with the literature, as the recorded average harvesting age for white pines (sawtimber) is 40 - 80 years (\citet{wittwer2004even}).
As expected, the agent would harvest earlier in the conservative case while postponing it in the optimistic case. Figure \ref{fig: sample path} (a) displays a representative sample path of the simulated forward process $X^\pi$ and the corresponding approximate optimal harvesting time ($\tauh^*=54.4$ years) with no uncertainty, shown in blue.
\begin{table}[H]
\centering
\begin{tabular}{cccc}
\hline
Case                    & Conservative & No uncertainty & Optimistic \\ \hline
Optimal harvesting time & 51.90        & 52.95          & 56.23      \\ \hline
\end{tabular}
\caption{{\footnotesize Approximate optimal harvesting time (in years).}}
\label{table: approximate optimal harvesting time}
\end{table}
As discussed in Section \ref{sec: uncertainty set}, each choice of uncertain parameters leads to a different risk-neutral probability measure, representing the agent's belief. Each probability measure reflecting the agent's beliefs results in a different optimal harvesting time. Therefore, we plot all the possible harvesting times as an interval for this specific sample path, which we refer to as the ``uncertainty corridor'' in Figure \ref{fig: sample path} (a). Figure \ref{fig: sample path} (b) shows the corresponding subjective forest lease value processes for the same path under three different scenarios: the conservative case ($v^-$), no uncertainty ($v$), the optimistic case ($v^+$) and the instant harvesting revenue. The subjective lease value process for the optimistic case ($v^+$) is always higher than or equal to that for the no uncertainty case ($v$), which in turn is always higher than or equal to the value process for the conservative case ($v^-$). Note, this reflects the subjective beliefs and values, and not market values, which we will consider later.  In all cases, the lease value processes are always higher than or equal to the instant harvesting revenue. The first time the lease value process touches the instant harvesting revenue is the optimal harvesting time. The corresponding approximate optimal harvesting time for each case is also highlighted: 53.35 years in the conservative case and 55.35 years in the optimistic case.
\begin{figure}[H]
  \centering
  \begin{tabular}{cc}
    \includegraphics[scale=0.55]{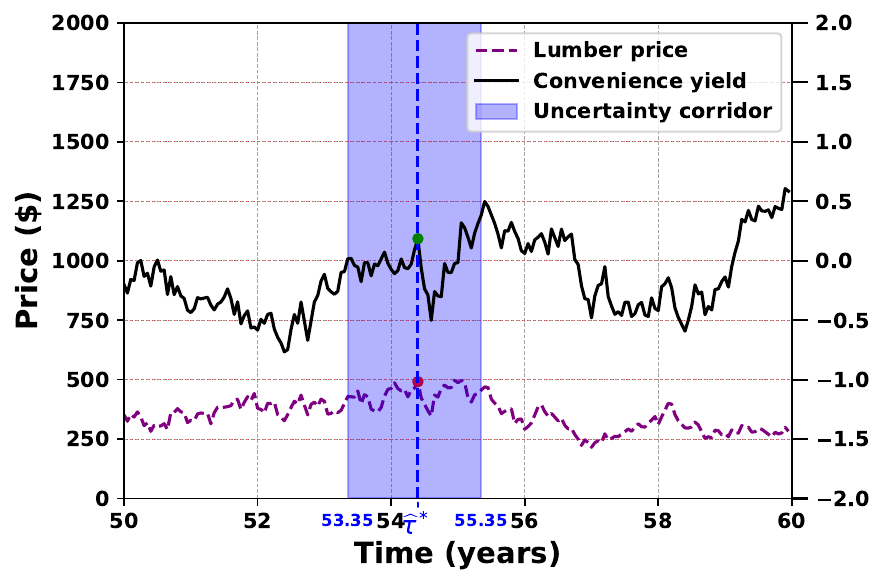} &
    \includegraphics[scale=0.55]{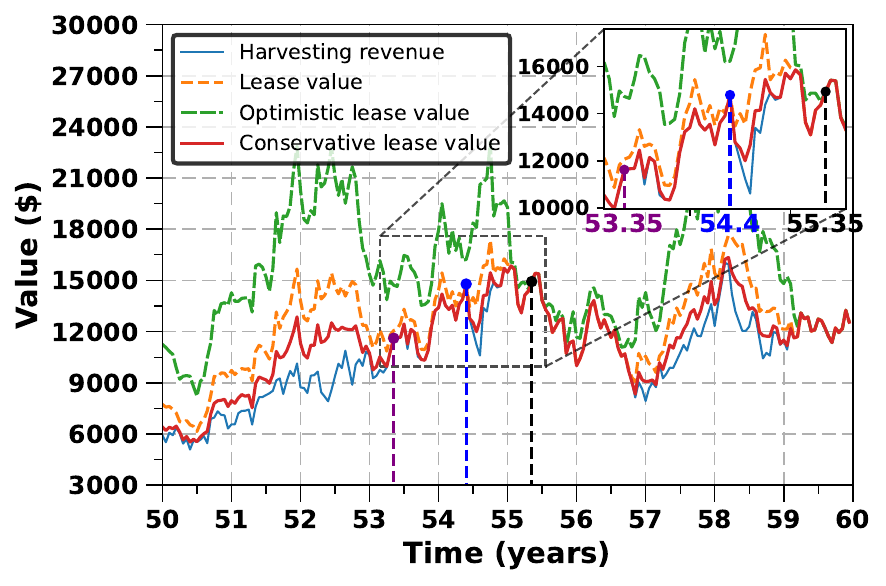}\\
    \footnotesize (a) Estimated optimal harvesting time &
    \footnotesize (b) Estimated forest lease value processes \\
  \end{tabular}
  \caption{{\footnotesize (a) The approximate optimal harvesting time of a sample path of $X$, and (b) the corresponding estimated forest lease value processes from 50 to 60 years.}}
  \label{fig: sample path}
\end{figure}

Since the SRODP algorithm locally approximates the continuation value as functions of $X^\pi_{t_i}$ at time $t_i$ in each hypercube $\Hc_j$ rather than globally, we can utilize the stored basis function regression coefficients for each $\Hc_j$ to efficiently recover the lease value at time $t_i$ as a function of the state variables $X^\pi_{t_i}$ by taking the maximum of the continuation value and the instant harvesting revenue. Figure \ref{fig: optimal stopping boundary 2} is a plot of both the instant harvesting revenue and lease value at age $70$ as a function of the lumber spot price and convenience yield in both conservative and optimistic cases. The overlapping regions of both surfaces, representing the stopping regions, are highlighted in black. We observe that the subjective lease value in the optimistic case ($v^+$) is slightly higher, while the overlapping region is slightly smaller than in the conservative case. This suggests that harvesting is more likely under the conservative scenario. Moreover, for a given lumber spot price, a lower convenience yield results in a higher forest lease value. Intuitively, a lower convenience yield implies that leaving the trees ``on the stump'' is more advantageous than holding harvested lumber inventories, as the storage costs outweigh the benefits of holding lumber. A similar observation was made by \citet{chen2012impact} in the case of no parameter uncertainty.
\begin{figure}[H]
  \centering
  \begin{tabular}{cc}
    \includegraphics[scale=0.52]{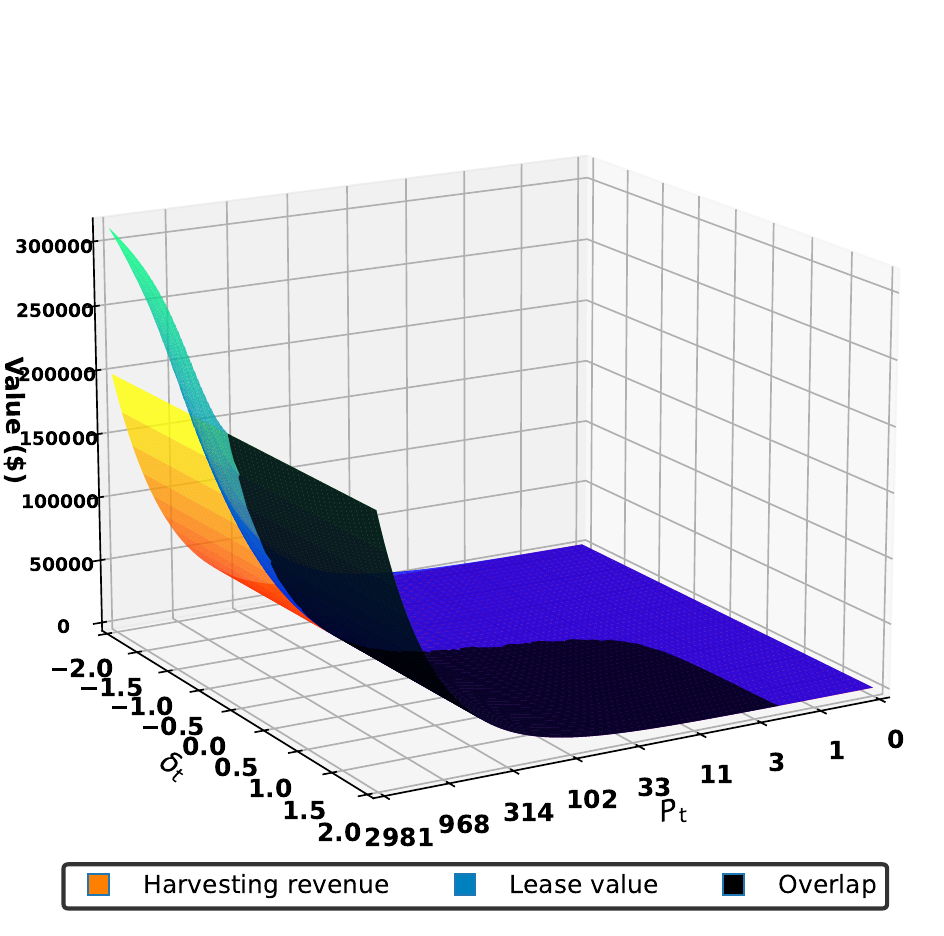} &
    \includegraphics[scale=0.52]{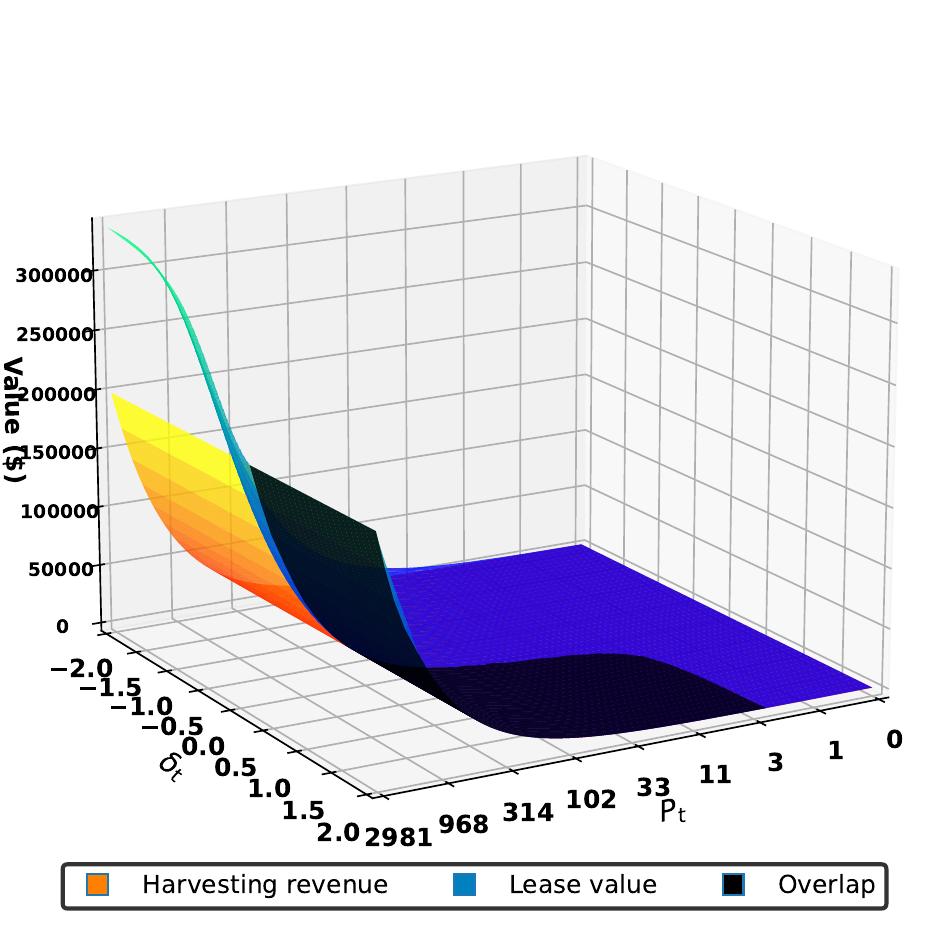}\\
    \footnotesize (a) Conservative case &
    \footnotesize (b) Optimistic case \\
  \end{tabular}
  \caption{{\footnotesize The approximate forest lease value and instant harvesting revenue as functions of the state variables at age 70.}}
  \label{fig: optimal stopping boundary 2}
\end{figure} 

To better understand the relationship between the optimal stopping boundary and the state variables, we project the instant harvesting revenue and subjective lease value onto the $P_t$-$\delta_t$ plane at various times across different scenarios, as shown in Figure \ref{fig: optimal stopping boundary 1}. The stopping boundary is a function of time, the lumber spot price, and the convenience yield. In all three cases, the stopping region expands as the time progresses toward the maturity date, suggesting an increasing possibility of early stopping or harvesting. Notably, the convex shape of stopping regions here resembles that of American type options, whose payoff depends on the product of two asset prices (for a detailed study, see \citet{broadie1997valuation}). 

More specifically, the stopping boundaries under all three scenarios in Figure \ref{fig: optimal stopping boundary 1} exhibit a similar pattern - the stopping region expands along the negative x- and y-axes as time progresses, indicating that lower lumber spot prices and convenience yields can trigger early harvesting. Neither the lumber spot price nor the convenience yield alone is a sufficient statistic for determining whether instant harvesting is optimal. However, for a given lumber spot price, a higher convenience yield is more likely to trigger early harvesting. The logic here is that a higher convenience yield reflects greater benefits of holding lumber inventories relative to storage costs, making it more beneficial for the agent to harvest the forest and store the lumber rather than leaving the trees ``on the stump.'' Furthermore, given that lumber spot prices typically range in the hundreds to over a thousand dollars per 1000 board feet (cf. average futures prices in Table \ref{table: futures data}), the convenience yield becomes determinant in the optimal harvesting policy. Thus, when lumber spot prices are moderate, it may be optimal to harvest even if storage costs slightly exceed the benefits of holding lumber inventories, corresponding to a slightly negative convenience yield.

It can be further observed that a conservative probability belief accelerates the optimal harvesting time, while an optimistic belief delays it. In particular, with a moderate lumber spot price, it could be optimal under the conservative scenario to harvest at or after age 70, even if the storage cost slightly outweighs the benefit of holding lumber inventories. Conversely, in the optimistic scenario, harvesting is delayed, becoming optimal only if the holding benefits are significantly higher than the storage costs, corresponding to a positive convenience yield.

\begin{figure}[H]
  \centering
\begin{tabular}{>{\centering\arraybackslash}m{1.5cm}>{\centering\arraybackslash}m{4.5cm}>{\centering\arraybackslash}m{4.5cm}>{\centering\arraybackslash}m{4.5cm}}
        & Conservative case & No uncertainty & Optimistic case \\
        \parbox{1.5cm}{\centering $t=50$} & \includegraphics[width=0.28\textwidth]{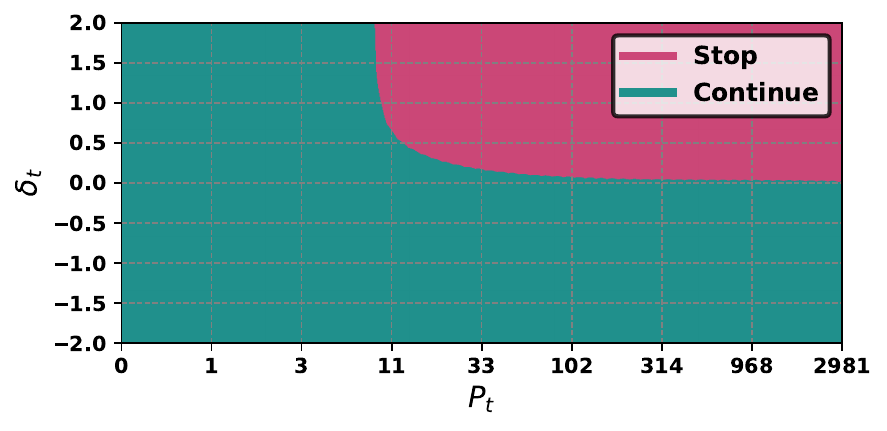} & \includegraphics[width=0.28\textwidth]{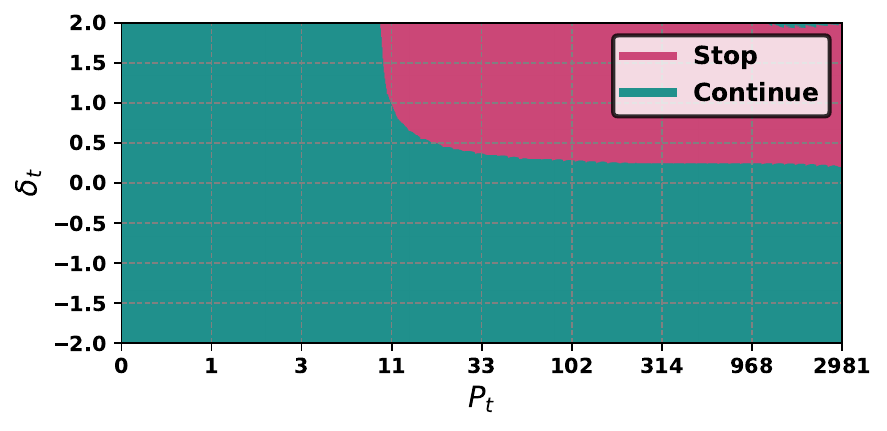} & \includegraphics[width=0.28\textwidth]{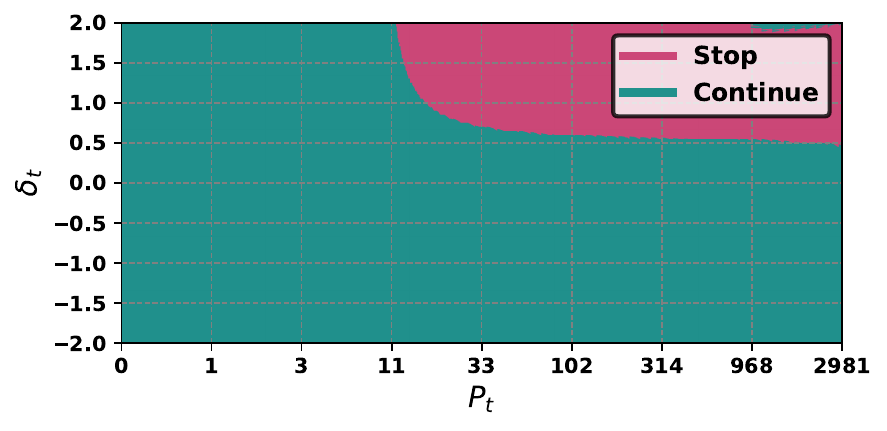} \\
        \parbox{1.5cm}{\centering $t=70$} & \includegraphics[width=0.28\textwidth]{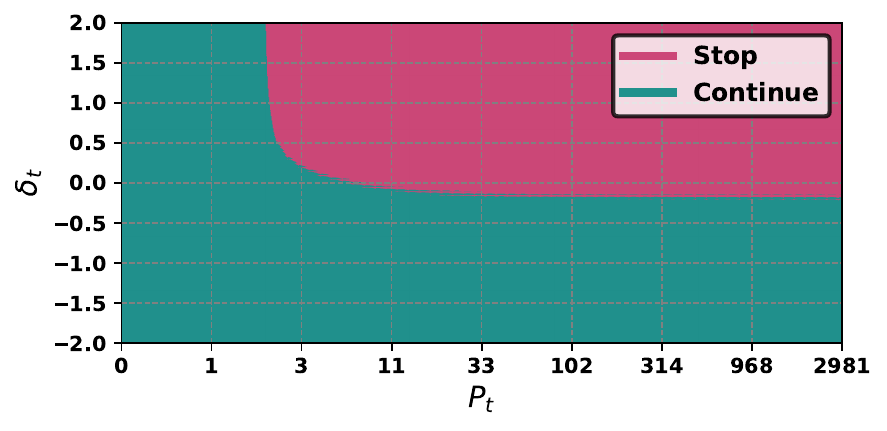} & \includegraphics[width=0.28\textwidth]{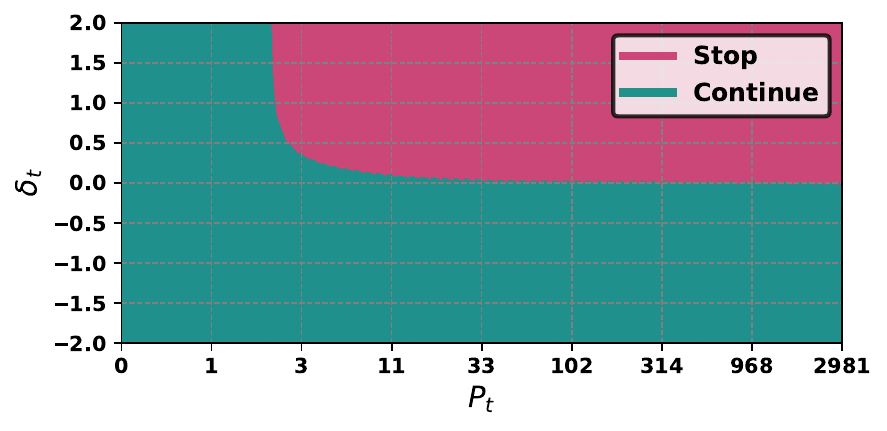} & \includegraphics[width=0.28\textwidth]{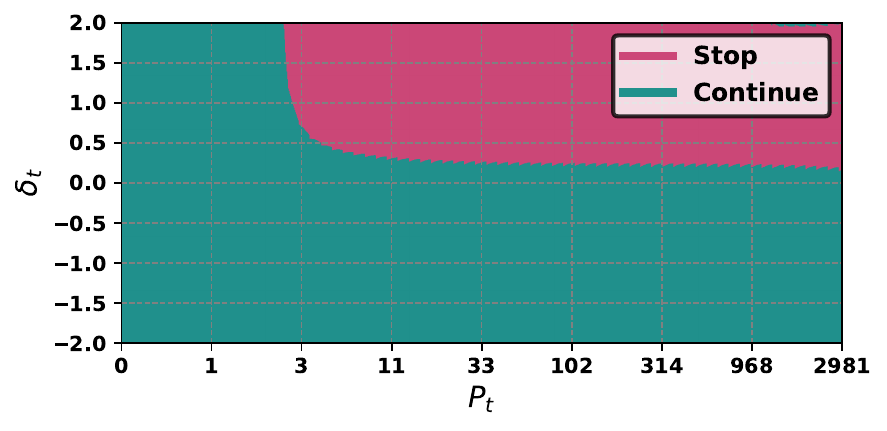} \\
        \parbox{1.5cm}{\centering $t=100$} & \includegraphics[width=0.28\textwidth]{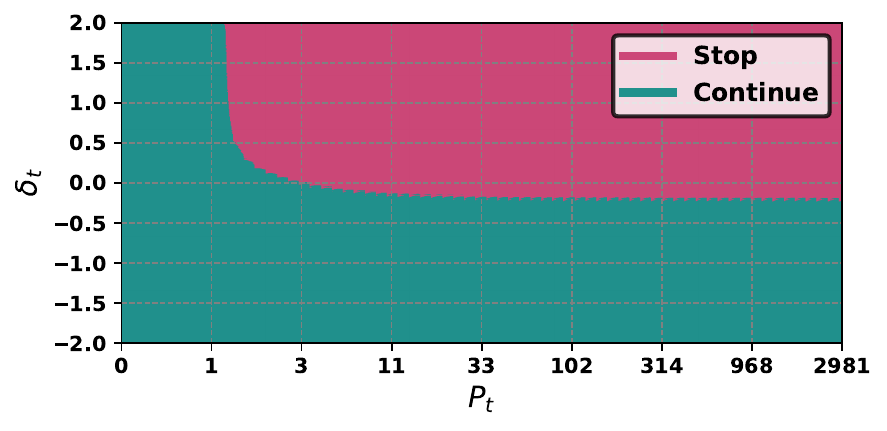} & \includegraphics[width=0.28\textwidth]{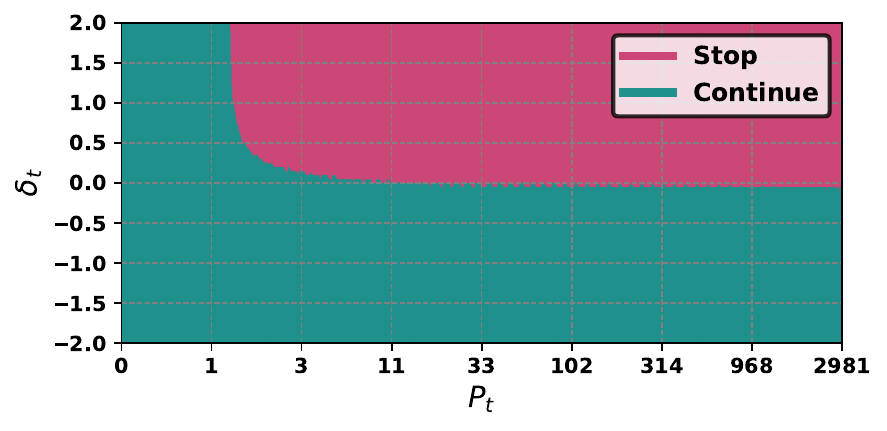} & \includegraphics[width=0.28\textwidth]{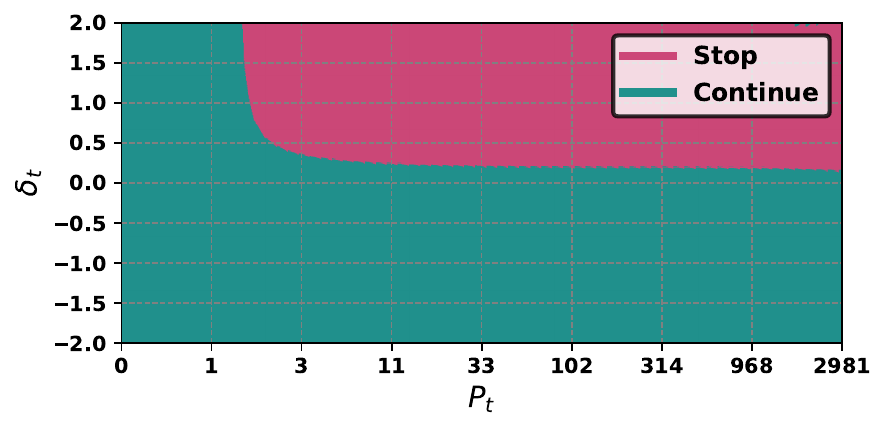} \\
\end{tabular}
\caption{{\footnotesize Stopping and continuation regions for the forestry investment.}}
\label{fig: optimal stopping boundary 1}
\end{figure}

\subsubsection{The case of pure catastrophe jump intensity uncertainty}\label{sec: with only catastrophe jump intensity uncertainty}
There is a general consensus that the frequency of catastrophe arrivals has increased over the last decades due to global warming and climate change. The limited availability of related data underscores the importance of closely examining uncertainty in the catastrophe jump intensity. Abstracting from other issues, considering the scenario where uncertainty exists solely in the catastrophe jump intensity might be of particular interest. We re-estimate the optimal harvesting time by averaging the estimated stopping times for $10^5$ re-simulated paths, with results presented in Table \ref{table: approximate optimal harvesting time only cat uncertainty}. It can be observed that, in the conservative case, optimal harvesting is slightly delayed, while in the optimistic case, it is slightly accelerated compared to the results in Table \ref{table: approximate optimal harvesting time}. This is intuitive, as the uncertainty set here is smaller, resulting in the minimum and maximum values of $v^u$ being higher and lower than those in Section \ref{sec: harvesting strategy with full parameter uncertainty}, respectively. The stopping and continuation region plots are nearly identical to those in Figure \ref{fig: optimal stopping boundary 1} and are therefore not presented here. 
\begin{table}[H]
\centering
\begin{tabular}{ccc}
\hline
Case                    & Conservative & Optimistic \\ \hline
Optimal harvesting time & 51.94                & 55.86      \\ \hline
\end{tabular}
\caption{{\footnotesize Approximate optimal harvesting time (in years) with only catastrophe jump intensity uncertainty.}}
\label{table: approximate optimal harvesting time only cat uncertainty}
\end{table}

\subsubsection{With carbon sequestration value}\label{sec: with carbon sequestration value}
Since living trees absorb carbon dioxide (CO2) from the atmosphere, delaying harvesting forest lands can help offset greenhouse gas emissions. Moreover, many corporations are seeking to purchase carbon credits in carbon markets to improve their environmental, social, and corporate governance (ESG) ratings. On the Voluntary Carbon Market in the USA, the average value of delaying harvest on private forest land for one year is around \$ 16 per acre (\citet{kreye2023}). We refer to this value of delaying harvest as the carbon sequestration value. By incorporating the carbon sequestration value into the amenity value ($A = 47.54$ \$/hectare), we re-estimate the optimal harvesting time by repeating the procedure used for the full parameter uncertainty case in Section \ref{sec: harvesting strategy with full parameter uncertainty}. The results presented in Table \ref{table: approximate optimal harvesting time with carbon sequestration value} indicate that including the carbon sequestration value delays the approximate optimal harvesting time by about 340 days, 468 days and 654 days under conservative, no uncertainty and optimistic beliefs, respectively. This result is expected thanks to the benefit of delaying harvest brought by the carbon sequestration value. The effect of delaying harvest is further demonstrated by the shrinking stopping regions observed in Figure \ref{fig: optimal stopping boundary 1 carbon} when compared with Figure \ref{fig: optimal stopping boundary 1}.
\begin{table}[H]
\centering
\begin{tabular}{cccc}
\hline
Case                    & Conservative & No uncertainty & Optimistic \\ \hline
Optimal harvesting time & 52.83        & 54.23          & 58.02      \\ \hline
\end{tabular}
\caption{{\footnotesize Approximate optimal harvesting time (in years) when including the value of delaying harvest.}}
\label{table: approximate optimal harvesting time with carbon sequestration value}
\end{table}

\begin{figure}[H]
  \centering
\begin{tabular}{>{\centering\arraybackslash}m{1.5cm}>{\centering\arraybackslash}m{4.5cm}>{\centering\arraybackslash}m{4.5cm}>{\centering\arraybackslash}m{4.5cm}}
        & Conservative case & No uncertainty & Optimistic case \\
        \parbox{1.5cm}{\centering $t=50$} & \includegraphics[width=0.28\textwidth]{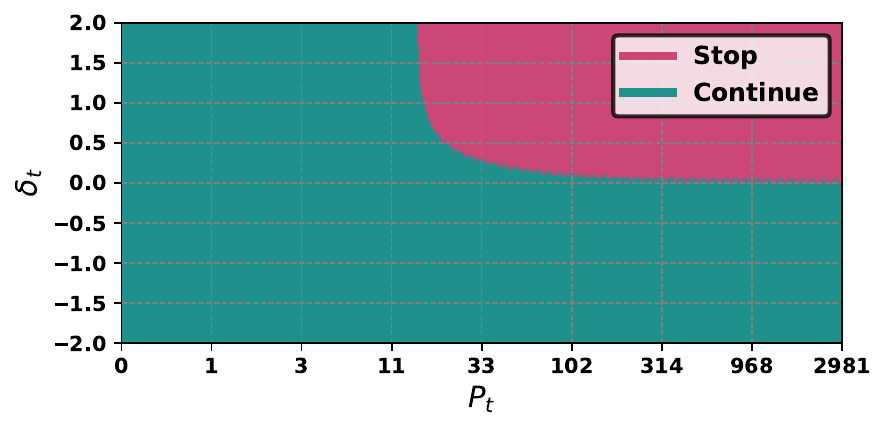} & \includegraphics[width=0.28\textwidth]{plots/stopping_boundary_1_t=50_no_amb.pdf} & \includegraphics[width=0.28\textwidth]{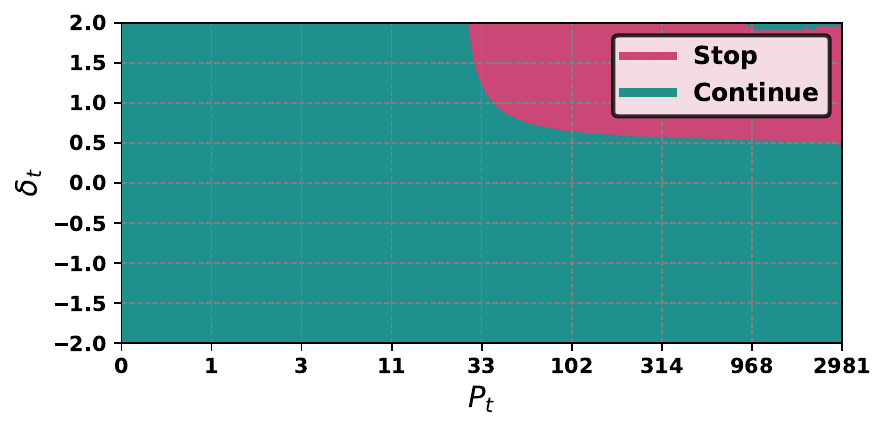} \\
        \parbox{1.5cm}{\centering $t=70$} & \includegraphics[width=0.28\textwidth]{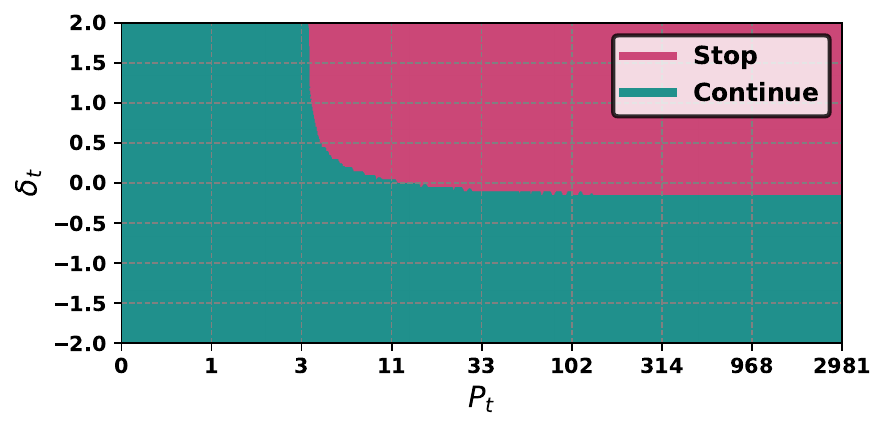} & \includegraphics[width=0.28\textwidth]{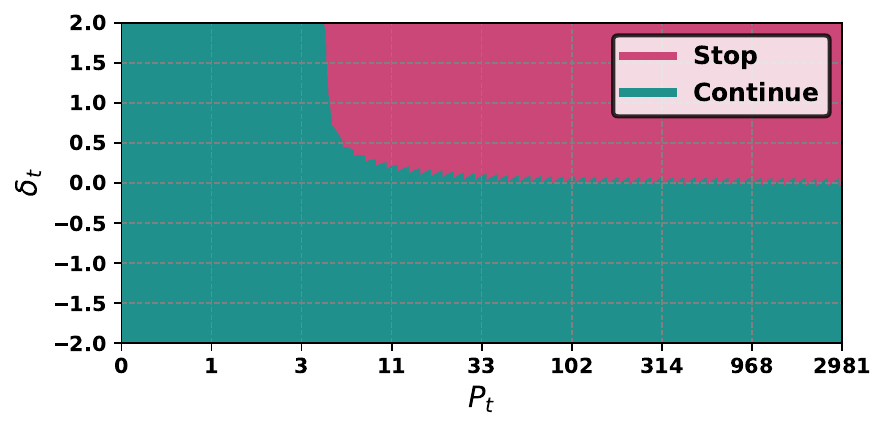} & \includegraphics[width=0.28\textwidth]{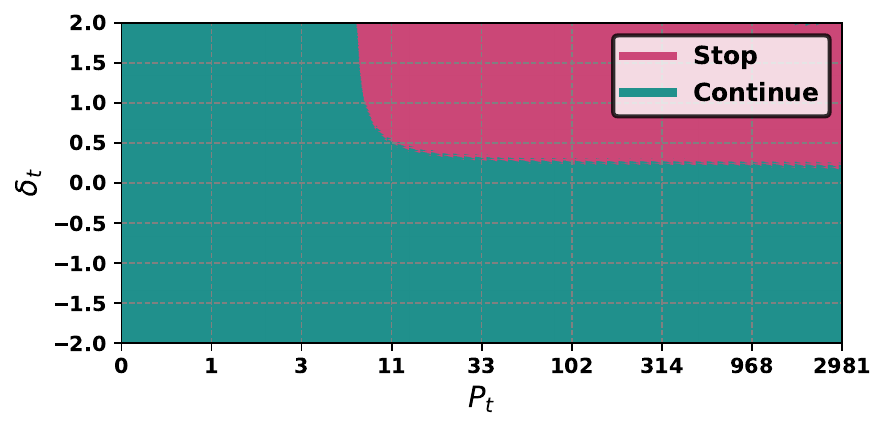} \\
        \parbox{1.5cm}{\centering $t=100$} & \includegraphics[width=0.28\textwidth]{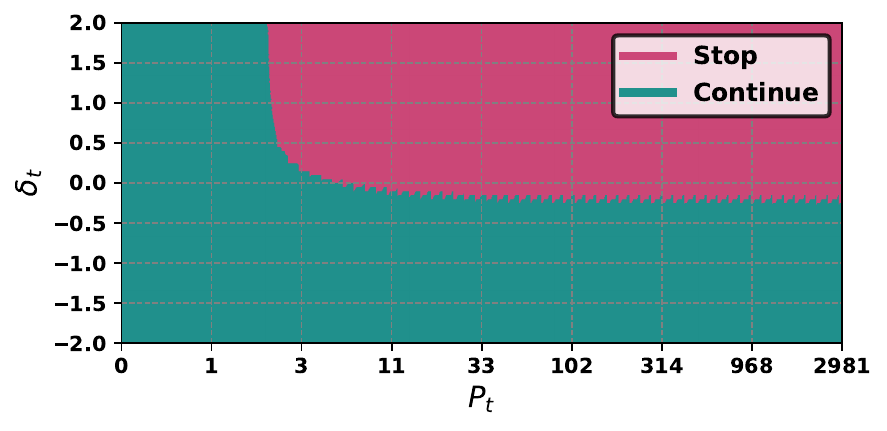} & \includegraphics[width=0.28\textwidth]{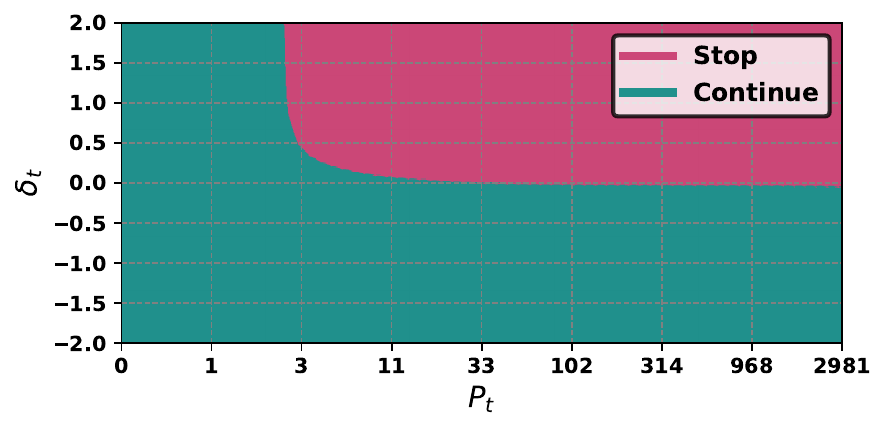} & \includegraphics[width=0.28\textwidth]{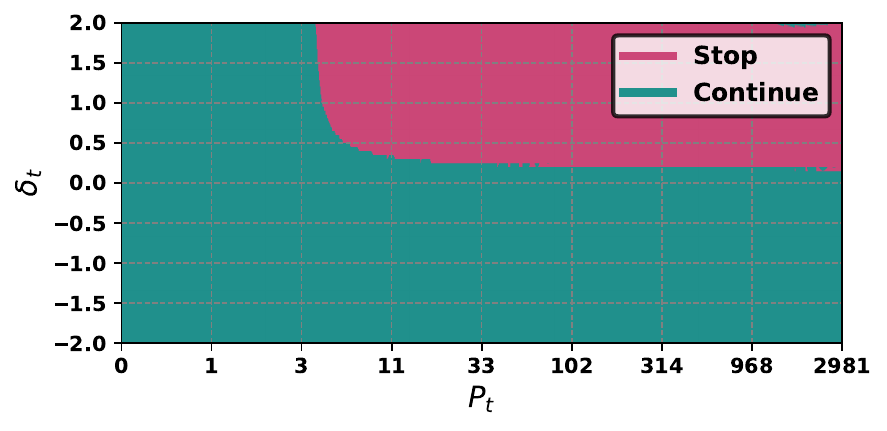} \\
\end{tabular}
\caption{{\footnotesize Stopping and continuation regions for the forestry investment with carbon sequestration value.}}
\label{fig: optimal stopping boundary 1 carbon}
\end{figure}

\subsection{Forest lease values under full parameter uncertainty}\label{sec: Forest lease values under full parameter uncertainty}
Let us assume that the estimated parameter values in Section \ref{sec: data summary} represent the ``true'' market measure $\Qb$. Under the no-arbitrage assumption, all asset prices, including the ``true'' lease value, have to be computed under the measure $\Qb$, and lease values ($v$) obtained in this way generally differ from the subjective lease values ($v^u$) obtained from \eqref{eq: F-reduced Qu value under Q}. The optimal harvesting strategy under the ``true'' no uncertainty scenario is $\tauh$ and the harvesting strategies $\tauh^+$ and $\tauh^-$ reflecting subjective beliefs under the uncertainty assumption can be assessed against $\tauh$. By computing the lease values corresponding to $\tauh^+$ and $\tauh^-$, respectively, under the pricing measure $\Qb$ and comparing with the lease value $v$ of $\tauh$ we obtain a measure of the cost that uncertainty generates for the forestry management. More specifically, we first derive approximate optimal stopping rules $\tauh^+$ and $\tauh^-$ for agents considering parameter uncertainty and then evaluate the forest lease values for these agents under $\Qb$ by applying their corresponding stopping rules as in \eqref{eq:lease value est}. We use $10^3$ re-simulated sample paths under $\Qb$ in the numerical method detailed in Appendix \ref{appendix: algo} to evaluate the forest lease values with the estimated stopping rule. We run the algorithm 10 times and compute the average and standard deviation of the estimated forest lease values. The average GPU computational time per run is about 130 seconds. We compute the difference in value (DIV) with respect to the lease value estimate when no parameter uncertainty is considered. Our proposed algorithm provides an efficient way to quantify the possible profit and loss in the trade of forest leases due to uncertainty. For this purpose, we define DIV as $\tfrac{z_2-z_1}{z_1}$, where $z_1$ stands for the forest lease value under no parameter uncertainty, obtained under $\Qb$, and $z_2$ is the SRODP estimate with parameter uncertainty. We present the results in Table \ref{table: valuation results}. It can be observed that the forest lease value is depreciated in both the optimistic and conservative cases. Interestingly, the optimistic case exhibits a much more significant decrease, with an estimated loss of approximately 15\%, while an agent with a conservative belief incurs only very minor losses. Intuitively this reflects the rule, ``when in doubt, better to be on the conservative side.'' It also provides clear policy advice, that in the case of parameter uncertainty, forestry valuation and harvesting policies should be determined from the assumption of conservative beliefs. However, more generally the findings highlight the substantial impact of parameter uncertainty on forest lease valuation.
\begin{table}[H]
\centering
\begin{tabular}{ccccc}
\hline
Case           & SRODP   & SRODP & 95\% Confidence        & DIV    \\
               & Est     & Std   & Interval               & (\%)   \\ \hline
No uncertainty & 1637.08 & 99.14 & {[}1575.63, 1698.53{]} & 0      \\ \hline
Optimistic     & 1394.86 & 92.53 & {[}1337.51, 1452.21{]} & -14.80  \\ \hline
Conservative   & 1628.06 & 103.84 & {[}1563.70, 1692.42{]} & -0.55 \\ \hline
\end{tabular}
\caption{{\footnotesize Lease values (in dollars) of one-hectare pine forest land under the two-factor stochastic convenience yield model.}}
\label{table: valuation results}
\end{table}

\section{Conclusion}\label{sec: conclusion}
We studied the lease valuation and optimal harvesting strategies, key instruments for forestry investment, within a framework involving stochastic financial models and catastrophe risk while accounting for model parameter uncertainty. Catastrophe arrivals were modeled using a Poisson point process. We used a two-factor stochastic convenience yield model to describe the lumber price dynamics, estimated via Kalman filtering and maximum likelihood estimation using lumber futures data. The Poisson jump intensity for catastrophe arrivals was estimated by facilitating major wildfire data for Douglas County, Oregon, U.S. The parameter uncertainty set was quantified by confidence regions derived from statistical inference on these parameter estimates. After that, we derived a theoretical representation of forest lease values as solutions to reflected backward stochastic differential equations (RBSDEs) under model uncertainty. Finally, we examined the effects of statistical parameter uncertainty and the inclusion of carbon sequestration value on forest lease values through Monte Carlo experiments. We also analyzed the impact of uncertainty on optimal harvesting times and stopping boundaries under optimistic, no uncertainty, and conservative scenarios using numerical experiments. Our results indicate that an optimistic probability belief delays optimal harvesting, while a conservative probability belief accelerates it. But neither are ``truely'' optimal. By benchmarking against the true optimal harvesting rule, we assess the cost that uncertainty imposes on forestry management, an element that will become increasingly important in the context of climate change.  We conclude that in the presence of parameter uncertainty, it is better to lean toward a conservative strategy reflecting, to some extent, the worst case than being overly optimistic.
Finally, we investigated how acknowledging the potential value of carbon sequestration changes harvesting and lease value, an element that becomes increasingly important in the process of the energy transition.  

\appendix

\section{Numerical algorithm}
\label{appendix: algo}
We adapt the Stratified Regression One-step Forward Dynamic Programming (SRODP) algorithm of \citet{agarwal2023monte} to approximately solve the RBSDE in \eqref{eqSRODP1}-\eqref{eqSRODP3}. The main idea is to obtain an accurate approximation of the conditional expectation in \eqref{eqSRODP1}-\eqref{eqSRODP2} which is done by leveraging the classical idea of stratified sampling (see \citet{gobet2016stratified}). In this algorithm, we first set the space domain $[E^1,E^2]\times[E^3,E^4]$ for simulating the values of $X^{\pi}_{t_i}$ at each $t_i$ with $E^1,E^2,E^3$ and $E^4$ chosen appropriately. This approach is called stratified sampling where the initial sample value is uniformly sampled within a specified range. For achieving higher accuracy, we further partition $[E^1,E^2]$ into $J$ disjoint cells and $[E^3,E^4]$ into another $J$ disjoint cells, each of equal width. These partitions create a total of $J^2$ hypercubes $\Hc_j \subset \Rb^2:\, 1\leq j\leq J^2$. Then for each hypercube $\Hc_j$, we simulate $M$ independent samples of $X^{\pi}_{t_i}$ from a logistic distribution. Next we perform least-squares Monte Carlo basis function regressions to approximate the conditional expectations of the form $\Eb^{\Qb}(\cdot|X^{\pi}_{t_i})$ in each hypercube. We denote the $m$-th simulated sample of $X^{\pi}_{t_i}$ as $X^{(m,j)}_{t_i}$ and use the following function representation for $(Y^{\pi}, Z^{\pi})$:
$y_{t_i}(X^{\pi}_{t_i}) = Y^{\pi}_{t_i}, \,\, z_{t_i}(X^{\pi}_{t_i}) = Z^{\pi}_{t_i}, \, i=0,...,N,$
and their basis function regression approximation in the hypercube $\Hc_j$ as $\big(y^{M,j}_{t_i}(\cdot), z^{M,j}_{t_i}(\cdot)\big),$ respectively. Once we have numerically solved the RBSDE, we compute the approximate optimal strategies in \eqref{eq:est opt stop time} and \eqref{eq:est opt stop time nounc}. Finally, we use the Monte Carlo estimator to compute the forest lease valuations in \eqref{eq:lease value est} based on different optimal stopping strategies. The complete numerical procedure is given as follows:

\begin{itemize}
\item \textbf{Step 1:}  Initialization $(i=N):$ Set $y_{t_N}(\cdot)$ = $S(t_N,\cdot)$
\item \textbf{Step 2:} For $i = N-1, \ldots, 1,$ and for {$j\in\{1,...,J^2\}$}: Generate $M$ samples $(X^{(m,j)}_{t_i})_{1\leq m\leq M}$ of $X^{\pi}_{t_i}$ with logistic distribution and $M$ samples $\big(\Delta \Wt^{(m,j)}_{t_i}\big)_{1\leq m\leq M}$ with normal distribution $\mathcal{N}(0,\Delta t)$. Simulate $M$ paths of the forward process $X^{\pi}$ in \eqref{eq: discrete state vector} for one step from $t_i$ to $t_{i+1}$.
\item \textbf{Step 3:} Within each hypercube $\Hc_j$, estimate regression coefficients $b^{M,j}_{z,t_i}$ and $b^{M,j}_{y,t_i}$ with basis function $\phi(\cdot)$, and compute $z^{M,j}_{t_i}(\cdot)$ and $y^{M,j}_{t_i}(\cdot)$ (either sequentially or in parallel across $j$).
\begin{itemize}
    \item[a.] Estimate $b^{M,j}_{z,t_i}$ and compute $z^{M,j}_{t_i}(\cdot)$:
    \eqlnostar{SRODPstep3-1}{
    & b^{M,j}_{z,t_i} = \arginf_{b}\frac{1}{M}\sum\nolimits_{m=1}^M\Big|b\cdot\phi\Big(X_{t_i}^{(m,j)}\Big) - \frac{1}{\Delta t}\mathbb{E}^{\Qb}\Big[y^{M,j}_{t_{i+1}}\Big(X^{(m,j)}_{t_{i+1}}\Big)\Delta \Wt^{(m,j)}_{t_i}|X_{t_i}^{(m,j)}\Big]\Big|^2,\\
    \label{SRODPstep3-2}
    & z^{M,j}_{t_i}\Big(X^{(m,j)}_{t_i}\Big) = b^{M,j}_{z,t_i}\cdot \phi\Big(X^{(m,j)}_{t_i}\Big),
    }
    \item[b.\label{alg: extra step}] For computing $f^+$ calculate maximum of $f(t_i,X_{t_i}^{(m,j)},y^{M,j}_{t_{i+1}},z^{M,j}_{t_{i}}(X^{(m,j)}_{t_{i}}),u)$ with respect to $u$ over the eight corners of $U =[\underline{\kappa},\overline{\kappa}]\times[\underline{\mu},\overline{\mu}]\times[\underline{\lambda},\overline{\lambda}]$,
    where $f(t,x,y,z,u)$ is explicitly represented in \eqref{eq: rbsde Qu 1}. 
    For $f^-$ calculate the minimum of $f(t_i,X_{t_i}^{(m,j)},y^{M,j}_{t_{i+1}}(X^{(m,j)}_{t_{i+1}}),z^{M,j}_{t_{i}}(X^{(m,j)}_{t_{i}}),u)$ with respect to $u$ over the eight corners of $U$.
    
    \item[c.] Estimate $b^{M,j}_{y,t_i}$ and compute $y^{M,j}_{t_i}(\cdot)$:
    \eqlnostar{}{
    b^{M,j}_{y,t_i} & = \arginf_{b}\frac{1}{M}\sum\nolimits_{m=1}^M\Big|b\cdot\phi\Big(X_{t_i}^{(m,j)}\Big) - \mathbb{E}^{\Qb}\Big[y^{M,j}_{t_{i+1}}\Big(X^{(m,j)}_{t_{i+1}}\Big) \\
    & + f^\pm\Big(t_i,X^{(m,j)}_{t_i}, y^{M,j}_{t_{i+1}}\Big(X^{(m,j)}_{t_{i+1}}\Big), z^{M,j}_{t_{i}}\Big(X^{(m,j)}_{t_{i}}\Big)\Big)\Delta t|X_{t_i}^{(m,j)}\Big]\Big|^2,\\
    y^{M,j}_{t_i} & \Big(X^{(m,j)}_{t_i}\Big) = \max\Big\{b^{M,j}_{y,t_i}\cdot \phi\Big(X^{(m,j)}_{t_i}\Big), S\Big(t_i,X_{t_i}^{(m,j)}\Big)\Big\}.
    }
\end{itemize}
\item \textbf{Step 4:} Simulate $M'$ paths of $X^{\pi}$ from $t_0$ to $t_N$. For the $m$-th path ($m=1,\ldots,M'$), estimate optimal stopping strategy by computing 
\eqstar{
\widetilde{\tau}^m := \min\left\{t_i, i=0,...,N: y^{M',j}_{t_i}\left(X_{t_i}^{m,j}\right) = S\left(t_i,X_{t_i}^{m,j}\right)\right\}.
}
\item \textbf{Step 5:} Compute lease value estimate using $M'$ samples of the forward process $X$ as
\eqstar{
\frac{1}{M'} & \sum^{M'}_{m=1} \bigg(B_{0,\widetilde{\tau}^m}^{-1}\left(P^m_{\widetilde{\tau}^m}G_{\widetilde{\tau}^m} - K\right)\Gamma^{\lambda}_{0, \widetilde{\tau}^m} + \int_0^{\widetilde{\tau}^m}B_{0,t}^{-1}\Gamma^{\lambda}_{0,t} A_{t} \dd t \bigg) \\
& = \frac{1}{M'}\sum^{M'}_{m=1}\left(e^{(50-\widetilde{\tau}^m)(r+\lambda^{\Qb})-50r}\left(P^m_{\widetilde{\tau}^m}G_{\widetilde{\tau}^m} - K\right) + \frac{A(e^{-50(r+\lambda^{\Qb})} - e^{-(r+\lambda^{\Qb})\widetilde{\tau}^m})}{r+\lambda^{\Qb}} + \frac{A(1 - e^{-50r})}{r}\right).
}
\end{itemize}

\section{Proofs}\label{appendix: proofs}

\begin{proof}[\textbf{Proof of Lemma \ref{lemma: Qu value}}]
\begin{enumerate}
    \item[(i)] Since the Radon-Nikodym derivative 
\eqlnostar{}{\eta_t^{\alpha,\psi} = \frac{\dd \Qb^u}{\dd \Qb}\Big|_{{\Gc}_t}}
in \eqref{eq: RN derivative 2} is defined with respect to the enlarged $\sigma$-field $\Gc_t$, we need to reduce it to $\sigma$-field $\Fc_t$ for our case. To this end, we write that, for every $t\in[0, T]$,
\eqlnostar{}{\frac{\dd \Qb^u}{\dd \Qb}\Big|_{{\Fc}_t} = \Eb^{\Qb}\left(\eta_t^{\alpha,\psi}\big|\Fc_t\right) = \Eb^{\Qb}\left(\eta_t^{\alpha}\eta_t^{\psi}\big|\Fc_t\right) = \eta_t^{\alpha}\Eb^{\Qb}\left(\eta_t^{\psi}\big|\Fc_t\right).}
Further, we have
\eqlnostar{}{\Eb^{\Qb}\left(\eta_t^{\psi}\big|\Fc_t\right) = \Eb^{\Qb}\left(\Ec_t\left(\int_{(0, \,\cdot]}\left(\psi_k - 1\right)\dd \Mt_k\right)\bigg|\Fc_{t}\right) = 1,}
where the second equality can be seen as a consequence of (or Lemma 3.14 in \citet{aksamit2017enlargement} with $\Pb=\Qb$, $\gamma = \psi - 1$ and $M = \Mt$) the martingale property of the Dol\'eans-Dade exponential, since the conditional cumulative distribution function of $\xi$ given $\Fc_{t}$ has the form $1-\exp(-\int_0^t \lambda^{\Qb}_k\dd k)$. Then the Radon-Nikodym derivative with respect to the filtration $\Fb$ satisfies:
\eqlnostar{eq: F-reduced Radon-Nikodym derivative}{\frac{\dd \Qb^u}{\dd \Qb}\Big|_{{\Fc}_t} = \eta_t^{\alpha}, \,\,\text{for}\,\,t\in[0, T].}
Thus the claimed equality \eqref{eq: F-reduced Qu value under Q} holds by combining \eqref{eq: F-reduced Qu value}, \eqref{eq: F-reduced Radon-Nikodym derivative} and the abstract Bayes formula.

\item[(ii)] It suffices to apply Proposition 7.1 in \citet{el1997reflected} with $\Gamma_t = B^{-1}_{0,t}\Gamma^u_{0,t}\eta^{\alpha}_t$, $Y_t = v^{\Qb^u}_t$, $\beta_t = -(r+\lambda^u_t)$, $\gamma_t = \alpha_t$, $\delta_t = A$, $\xi = S_T$, $S_t = S_t$ and $B_t = \Wt_t$ to obtain \eqref{eq: rbsde Qu 1}.
\end{enumerate}
\end{proof}

\begin{proof}[\textbf{Proof of Theorem \ref{thm: robust rbsde}}]
We prove the first result mainly based on the comparison theorem for RBSDEs (see \citet[Theorem 4.1]{el1997reflected}) and the second one follows similarly. First, since $f^+(t,x,y,z)\geq f(t,x,y,z,u)$ by definition, we have that $Y^+_t \geq Y^u_t$ for all $t\in[0,T]$ and all $u\in\Uc$ by the comparison theorem. Then we have $Y^+_t \geq v^+_t, \, \forall t\in[0,T]$.

Next, we prove that $Y^+_t \leq v^+_t, \, \forall t\in[0,T]$. Note that for any $\varepsilon>0,$ there exists a Borel measurable function $u^{\varepsilon}: [0,T]\times\Rb^2\times\Rb\times\Rb^2\rightarrow\Uc$ such that 
\eqlnostar{}{f^+(t,x,y,z) - \varepsilon\leq f\left(t,x,y,z,u^{\varepsilon}(t,x,y,z)\right),}
by the Borel measurable selection theorem (see, for example, \citet[Chapter 7]{bertsekas1996stochastic}). Let $Y^\varepsilon_t = v^\varepsilon_t, \, t\in[0,T]$ be the corresponding solution of the RBSDE with generator $f\left(t,x,y,z,u^{\varepsilon}(t,x,y,z)\right)$. Since $Y^+_t - \varepsilon(T-t)$ solves the RBSDE with generator $f^+(t,x,y,z) - \varepsilon$ and obstacle $S_t - \varepsilon(T-t)$, we have that for every $\varepsilon>0,$ $Y^+_t - \varepsilon(T-t)\leq v^\varepsilon_t\,\, \forall t\in[0,T]$ by the comparison theorem again. Therefore, $Y^+_t \leq v^+_t, \, \forall t\in[0,T]$, and the proof is complete.
\end{proof}
\begin{theorem}
\label{theorem 2}
For a fixed EMM $\Qb\in\Qc$, the value $V$ of a forestry investment defined in \eqref{eq: value function 1} satisfies: $\Ib_{\{\xi>t\}}V_t = v_t$ for $t\in[0,T],$ where $v_t$ is called the $\Fb$-reduced value of the forestry investment under $\Qb$, denoted by
\eqstar{v_t := \underset{\tau\in\Tc_{t,T}}{\sup}\, \Eb^{\Qb}\left(B^{-1}_{t,\tau}\Gamma^{\lambda}_{t,\tau}\left(P_{\tau}G_{\tau} - K\right) + \int_t^{\tau}B_{t,u}^{-1}\Gamma^{\lambda}_{t,u}A_u\dd u\Big|\Fc_t\right),}
where $\Gamma^{\lambda}$ is the conditional survival probability of $\xi$ under $\Qb$ measure:
\eqlnostar{}{\Gamma^{\lambda}_{s,t}:=\Qb(\xi>t|\xi>s,\Fc_s) = \exp\left(-\int_s^t \lambda^{\Qb}_u\dd u\right).}
\end{theorem}
\begin{proof}[\textbf{Proof of Theorem \ref{theorem 2}}]
Let the $\Fb$-stopping time $\tau\in\Tc_{t,T}$ be a harvesting strategy. The expected value $V_t^{\tau}$ of this strategy defined as
\eqlnostar{}{V_t^{\tau} := \Eb^{\Qb}\left(B_{t,\tau}^{-1}\left(P_{\tau}G_{\tau} - K\right)\Ib_{\{\tau<\xi\}} + \int_{t\wedge\xi}^{\tau\wedge\xi} B_{t,u}^{-1} A_u\dd u\bigg| \Gc_t\right)}
satisfies
\eqlnostar{eq: theorem 1 proof 1}{V_t^{\tau} & = \Eb^{\Qb}\left(B_{t,\tau}^{-1}\left(P_{\tau}G_{\tau} - K\right)\Ib_{\{\tau<\xi\}} + \Ib_{\{\tau<\xi\}}\int_t^{\tau} B_{t,u}^{-1} A_u\dd u + \Ib_{\{t<\xi\leq\tau\}}\int_t^{\xi} B_{t,u}^{-1} A_u\dd u\bigg| \Gc_t\right).}
Let us introduce two auxiliary processes $R^1$ and $R^2$ by setting
\eqlnostar{}{R^1_{t} = B_{0,t}^{-1}\left(P_{t}G_{t} - K\right),\quad\quad R^2_{s,t} = \int_s^{t}B^{-1}_{0,u}A_u\dd u, \,\,\text{for}\,\,0\leq s\leq t\leq T.}
Then $V_t^{\tau}$ can be represented as 
\eqlnostar{eq: theorem 1 proof 2}{V_t^{\tau} & = B_{0,t}\Eb^{\Qb}\left(R^1_{\tau}\Ib_{\{\tau<\xi\}} + \Ib_{\{\tau<\xi\}}R^2_{t,\tau} + \Ib_{\{t<\xi\leq\tau\}}R^2_{t,\xi}\Big| \Gc_t\right).}
For the first part on the right-hand side of \eqref{eq: theorem 1 proof 2}, since $\{\tau<\xi\}$ is $\Gc_{\tau}$-measurable hence $\Gc_{T}$-measurable, $\Ib_{\{\tau<\xi\}}R^1_{\tau}$ is $\Gc_{T}$-measurable. By applying Lemma 3.1 in \citet{elliott2000models} we know that
\eqlnostar{}{\Eb^{\Qb}\left(\Ib_{\{\tau<\xi\}}R^1_{\tau}\big| \Gc_t\right) = \Eb^{\Qb}\left(\Ib_{\{\tau<\xi\}}R^1_{\tau}\big| \Gc_t\right)\Ib_{\{\xi>t\}} & = \frac{\Eb^{\Qb}\left(\Ib_{\{\tau<\xi\}}R^1_{\tau}\Ib_{\{\xi>t\}}\big| \Fc_t\right)}{\Eb^{\Qb}\left(\Ib_{\{\xi>t\}}\big| \Fc_t\right)}\Ib_{\{\xi>t\}}\\
& = \frac{\Eb^{\Qb}\left(\Eb^{\Qb}\left(\Ib_{\{\tau<\xi\}}R^1_{\tau}\big|\Fc_{\tau}\right)\big| \Fc_t\right)}{\Eb^{\Qb}\left(\Ib_{\{\xi>t\}}\big| \Fc_t\right)}\Ib_{\{\xi>t\}}\\
& = \Eb^{\Qb}\left(R^1_{\tau}\cdot\frac{\Eb^{\Qb}\left(\Ib_{\{\tau<\xi\}}\big|\Fc_{\tau}\right)}{\Eb^{\Qb}\left(\Ib_{\{\xi>t\}}\big| \Fc_t\right)}\Bigg|\Fc_t\right)\Ib_{\{\xi>t\}}\\
\label{eq: theorem 1 proof 3}
& = \Eb^{\Qb}\left(R^1_{\tau}\Gamma^{\lambda}_{t,\tau}|\Fc_t\right)\Ib_{\{\xi>t\}},}
where the third equality follows the tower property, the fourth equality follows the $\Fc_{\tau}$-measurability of $R^1_{\tau}$ and the last is the definition of $\Gamma^{\lambda}$. Similarly, we have the following result for the second part on the right-hand side of \eqref{eq: theorem 1 proof 2}
\eqlnostar{eq: theorem 1 proof 4}{\Eb^{\Qb}\left(\Ib_{\{\tau<\xi\}}R^2_{t,\tau}\big| \Gc_t\right) = \Eb^{\Qb}\left(R^2_{t,\tau}\Gamma^{\lambda}_{t,\tau}|\Fc_t\right)\Ib_{\{\xi>t\}},}

For the third part on the right-hand side of \eqref{eq: theorem 1 proof 2}, note that $\Ib_{\{t<\xi\leq\tau\}}R^2_{t,\xi}$ is left-continuous in $\xi$. It is well-known that every $\Fb$-adapted left-continuous process is $\Fb$-predictable. Therefore, we can apply Proposition 3.4 in \citet{elliott2000models} to obtain
\eqlnostar{}{\Eb^{\Qb}\left(\Ib_{\{t<\xi\leq\tau\}}R^2_{t,\xi}\big| \Gc_t\right)\Ib_{\{\xi>t\}} = & -\Ib_{\{t<\xi\}}\Eb^{\Qb}\left(\int_t^{\infty}\Ib_{\{\tau\geq u\}}R^2_{t,u}\exp\left(-\int_t^u\lambda^{\Qb}_s\dd s\right)\lambda^{\Qb}_u \dd u\bigg| \Fc_t\right)\Ib_{\{\xi>t\}}\\
& + \Ib_{\{\xi\leq t\}}\Ib_{\{\tau\geq\xi\}}R^2_{t,\xi}\Ib_{\{\xi>t\}}\\
\label{eq: theorem 1 proof 5}
= & -\Eb^{\Qb}\left(\int_t^{\tau}R^2_{t,u} \Gamma^{\lambda}_{t,u}\lambda^{\Qb}_u \dd u\bigg|\Fc_t\right)\Ib_{\{\xi>t\}}.}
To complete the proof, note that $\Gamma^{\lambda}_{t,\cdot}$ is continuous and $R^2_{t,\cdot}$ is also a continuous process with bounded variation and $R^2_{t,t} = 0$, so with It\^o's Lemma we have:
\eqlnostar{}{R^2_{t,\tau}\Gamma^{\lambda}_{t,\tau} - \int_t^{\tau}R^2_{t,u}\dd \Gamma^{\lambda}_{t,u} = \int_t^{\tau}\Gamma^{\lambda}_{t,u}\dd R^2_{t,u} = \int_t^{\tau}B_{0,u}^{-1}\Gamma^{\lambda}_{t,u}A_u\dd u.}
Then, by combining \eqref{eq: theorem 1 proof 3}, \eqref{eq: theorem 1 proof 4} and \eqref{eq: theorem 1 proof 5} we obtain
\eqlnostar{}{V_t^{\tau} & = B_{0,t}\Eb^{\Qb}\left(R^1_{\tau}\Gamma^{\lambda}_{t,\tau} + \int_t^{\tau}B_{0,u}^{-1}\Gamma^{\lambda}_{t,u}A_u\dd u\Big|\Fc_t\right)\\
& = \Eb^{\Qb}\left(B^{-1}_{t,\tau}\Gamma^{\lambda}_{t,\tau}\left(P_{\tau}G_{\tau} - K\right) + \int_t^{\tau}B_{t,u}^{-1}\Gamma^{\lambda}_{t,u}A_u\dd u\Big|\Fc_t\right).}
Thus, take $\Ib_{\{\xi>t\}}V_t = \underset{\tau\in\Tc_{t,T}}{\sup}\, V_t^{\tau}$ and we obtain $\Ib_{\{\xi>t\}}V_t = v_t$.
\end{proof}

\singlespacing
\normalem
\printbibliography
\end{document}